\documentclass[10pt]{amsart}
\usepackage{amsmath, amsfonts, amsthm, amssymb}
\usepackage[a4paper, left=4cm, right=3cm, width=16cm, top=30mm, bottom=30mm]{geometry}
\usepackage{mathrsfs}
\usepackage{stmaryrd}
\usepackage{verbatim}
\usepackage{hyperref}
\usepackage{color}
\usepackage{blindtext}
\usepackage{bbm}
\usepackage{ulem}
\usepackage{graphicx}
\usepackage{multirow}
\usepackage{enumerate}
\usepackage{caption}
\usepackage{mathtools}
\mathtoolsset{showonlyrefs}
\usepackage{subcaption}
\reversemarginpar
\usepackage[textsize=tiny]{todonotes}

\newtheorem{theorem}{Theorem}[section]
\newtheorem{corollary}[theorem]{Corollary}
\newtheorem{lemma}[theorem]{Lemma}
\newtheorem{proposition}[theorem]{Proposition}

\theoremstyle{definition}
\newtheorem{definition}[theorem]{Definition}
\newtheorem{remark}[theorem]{Remark}

\numberwithin{equation}{section}

\newcommand{\RR}{\mathbb{R}}
\newcommand{\PP}{\mathbb{P}}
\newcommand{\BS}{\mathrm{BS}}
\newcommand{\D}{\mathrm{d}}
\newcommand{\Pf}{\mathrm{P}}
\newcommand{\Ir}{\mathrm{I}}

\newcommand{\NN}{\mathbb{N}}

\newcommand{\Ff}{\mathcal{F}}

\newcommand{\Cc}{\mathcal{C}}
\newcommand{\Hh}{\mathcal{H}}
\newcommand{\Oo}{\mathcal{O}}
\newcommand{\EE}{\mathbb{E}}
\newcommand{\Ss}{\mathcal{S}}

\newcommand{\Ll}{\mathcal{L}}

\newcommand{\eps}{\varepsilon}

\newcommand{\E}{\mathrm{e}}

\newcommand{\dx}{\partial_{x}}
\newcommand{\dxx}{\partial_{x}^2}
\newcommand{\Ddx}{\mathcal{D}_{x}}

\newcommand{\dy}{\partial_{y}}
\newcommand{\dyy}{\partial_{y}^2}

\newcommand{\LY}{\Ll_{Y}}
\newcommand{\Lg}{\left\langle}
\newcommand{\Rg}{\right\rangle}
\newcommand{\sigmaLim}{\boldsymbol{\varkappa}}
\newcommand{\half}{\frac{1}{2}}

\newcommand{\tbrK}{K}
\newcommand{\tbrKab}{K^a_\beta}
\newcommand{\tbrooK}{\overline{\overline{K}}}
\newcommand{\tbrooKab}{\overline{\overline{K}}^{\overline{a}}_\beta}

\newcommand{\tbrthKab}{\widetilde{\widehat{K}}^{\widehat{a}}_\beta}
\newcommand{\tbrhK}{\widehat{K}}
\newcommand{\tbrhKab}{\widehat{K}^{\widehat{a}}_\beta}
\newcommand{\tbrhhKab}{\widehat{\widehat{K}}^{\widehat{a}}_\beta}
\newcommand{\tbrepsK}{K_\eps}
\newcommand{\tbrtKyb}{\widetilde{K}^{y_0}_\beta}
\newcommand{\tbrKxyb}{K_\beta^{x-y_\sigma}}
\newcommand{\tbrchif}{\overline{\chi}_1}
\newcommand{\tbrchis}{\overline{\chi}_2}
\newcommand{\tbrchit}{\overline{\chi}_3}

\newcommand{\LBsLim}{\Ll_{\BS}^{\sigmaLim}}
\newcommand{\amin}{\alpha_{0}}
\newcommand{\dmin}{\delta_{0}}
\newcommand{\Dom}{\mathfrak{D}}

\newcommand{\Sigov}{\overline{\sum}}
\newcommand{\Sigund}{\underline{\sum}}

\begin{document}

\title{A theoretical analysis of Guyon's toy volatility model}

\author{Ofelia Bonesini}
\address{Department of Mathematics, University of Padova}
\email{bonesini@math.unipd.it}

\author{Antoine Jacquier}
\address{Department of Mathematics, Imperial College London, and Alan Turing Institute}
\email{a.jacquier@imperial.ac.uk}

\author{Chlo\'e Lacombe}
\address{Department of Mathematics, Imperial College London}
\email{chloe.lacombe14@imperial.ac.uk}

\date{\today}
\thanks{The authors would like to thank Alexander Kalinin for his insightful remarks about 
boundary behaviour of solutions of SDEs as well as Jean-Pierre Fouque for his pointers on ergodic diffusions,
and Yuliya Mishura for her hints about boundedness of moments.
They are also indebted to Julien Guyon for introducing them to this exciting problem.
AJ acknowledges financial support from the EPSRC grant EP/T032146/1.}

\maketitle
\begin{abstract}
We provide a thorough analysis of the path-dependent volatility model introduced by Guyon~\cite{G17},
proving existence and uniqueness of a strong solution, characterising its behaviour at boundary points,
providing asymptotic closed-form option prices as well as deriving small-time behaviour estimates. 
\end{abstract} 

\tableofcontents

\section{Introduction}
Stochastic volatility models have been used extensively over the past three decades in order to reproduce particular features of market data,
on Equities, FX and Fixed Income markets, both under the historical measure and for pricing purposes.
Most of them are based on a Markovian assumption for the underlying process,
essentially for mathematical convenience, as PDE techniques and Monte Carlo schemes are more readily available then. 
However, recent models have departed from this Markovian confinement and have shown to provide extremely accurate fit to market data.
One approach considers instantaneous volatility driven by fractional Brownian motion, giving rise to the rough volatility generation and its numerous descendants~\cite{Alos, BayerFriz, ElEuch, FZ17, Fuka, GatheralRough, Guennoun}.
A less strodden, yet very intuitive, path, originally introduced by Engle~\cite{E82} and Bollerslev~\cite{B86} in the early 1980s
suggested to consider models where volatility depends on the past history of the stock price process.
Their approach, though, was under the historical measure, and Duan~\cite{D95} investigated these discrete-time models in the context of option pricing.
With this in mind, Hobson and Rogers~\cite{HR98} extended this approach to continuous time, 
proposing  that instantaneous volatility should depend on exponentially weighted moments of the stock price. 
Contrary to stochastic volatility models, 
the market here is complete. 
Hobson and Rogers~\cite{HR98} showed that such models generate implied volatility smiles and skews consistent with market data. 
Further results investigated some theoretical properties of these models, in particular~\cite{MZ09} proving existence and uniqueness of strong solutions.
This path has recently been given new highlights by Guyon~\cite{Guyon14}, who assumed that the underlying stock price process behaves as
$$
\frac{\D S_t}{S_t} = \sigma(t, S_t, Y_t) l(t,S_t) \D W_t, \quad S_0 := s_0 >0,
$$
where~$W$ is a standard Brownian motion, $Y$ an adapted process and~$l(\cdot)$ a leverage function ensuring that European options are fully recovered.
Inspired by Hobson and Rogers~\cite{HR98}, Guyon~\cite{G17} suggested to choose~$Y$ as an exponentially weighted moving average of~$S$.
Not only does this model calibrate perfectly to the observed smile, but the diffusion map~$\sigma(\cdot)$ can be chosen in such a way that 
joint calibration with VIX data becomes feasible, a notoriously difficult task.

Motivated by his empirical results,
we investigate the theoretical properties of this model.
We provide a full characterisation of the behaviour of the volatility process at its boundaries, together with its ergodic behaviour, and derive closed-form asymptotics for the corresponding option prices in small time.
In Section~\ref{sec:SetUp}, we set the notations and present Guyon's model. 
Section~\ref{sec:Results} gathers the main theoretical results, proving  existence and uniqueness of a strong solution (Section~\ref{sec:Boundary}),
deriving the stationary distribution (Section~\ref{sec:Stationary}),
which we use to obtain an expansion of the option price in Section~\ref{sec:PDE}.
We finally provide small-time option price and implied volatility asymptotics for this model in Section~\ref{sec:LDP}. 
We gather all (lengthy) proofs in the appendix.


This project arises as an empirical analysis carried out by Guyon~\cite{G17} (see also~\cite{Guyon14})
to describe the relationship between the VIX index and the VVIX, a volatility of volatility index.
Figure~\ref{pic:VIX_VVIX} below shows a scatter plot of one versus the other over a five-year period.
The approximate linear relationship highlighted by the least-square regression fit was first noted by Guyon~\cite{G17},
and we follow his recommendations here.
\begin{figure}[h!]
\centering
\includegraphics[scale=0.5]{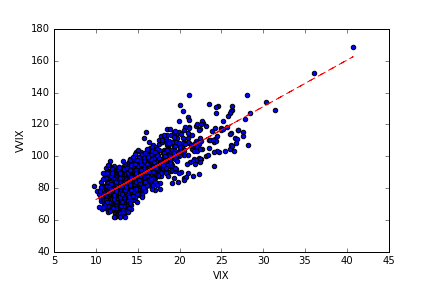}
\caption{Historical VVIX vs historical VIX (13/4/12-8/5/17).
\small{\textit{Source: CBOE.}}}
\label{pic:VIX_VVIX}
\end{figure}
\section{Set up and notations}\label{sec:SetUp}

The underlying process~$S$, describing the evolution of the S\&P index follows the general dynamics
$$
\frac{\D S_t}{S_t} = \sigma(Y_t)\D W_t, \quad S_0 = s_0 >0,
$$
for some given Brownian motion~$W$ generating a filtration~$\Ff=(\Ff_t)_{t\geq 0}$, 
where $\sigma:\RR^*_+\to \RR$ is non anticipative. 
Following Guyon~\cite{G17} and Hobson and Rogers~\cite{HR98}, 
we assume that the process~$Y$ is adapted to~$\Ff$
and is a function of the past history of the stock~$S$, 
making the latter non-Markovian, in the sense
\begin{equation}\label{eq:DefY}
Y_t := \frac{S_t}{\overline{S}^h_t}, \quad \text{for }t\in [0,T], 
\qquad \text{where }
\overline{S}^h_t := \frac{1}{h} \int_{-\infty}^t \exp\left\{- \frac{t-u}{h}\right\} S_u \D u
\end{equation}
is the exponentially weighted moving average (EWMA) of the stock price process. 
Here, the time horizon is set to be~$T$.
The constant~$h>0$, denoting the length of the time window, is left unspecified for now.
Using It\^o's formula and denoting $X:=\log(S)$, we can summarise the dynamics for the couple~$(X, Y)$ as
\begin{equation}\label{eq:SDE}
\left\{
\begin{array}{rll}
\D X_t & = \displaystyle -\half\sigma(Y_t)^2\D t + \sigma(Y_t) \D W_t,  & X_0 = x_0:=\log(s_0),\\
\D Y_t & = \displaystyle b(Y_t) \D t + \widetilde{\sigma}(Y_t) \D W_t, &Y_0 = y_0 >0,
\end{array}
\right.
\end{equation}
with $b(y) := \frac{1}{h}y(1-y)$ and $\widetilde{\sigma}(y) := y\sigma(y)$, for $y >0$
and some $h>0$. 
Guyon~\cite{G17} showed that, for the linear relationship between the VIX and the VVIX to hold, 
one needs to consider a diffusion coefficient of the form
\begin{equation}\label{eq:Sigma}
\sigma(y) := -\frac{\alpha}{\beta} + \gamma y^{-\beta},
\end{equation}
with $\alpha,\beta,\gamma >0$.
In that case, $\widetilde{\sigma}$ is null at $y_\sigma := \left( \frac{\beta\gamma}{\alpha}\right)^{1/\beta}$, and
\begin{equation*}
\widetilde{\sigma}(0) = \left\{
\begin{array}{ll}
\text{not defined}, & \text{if }\beta>1,\\
0, & \text{if }\beta<1,\\
\gamma, & \text{if }\beta=1.
\end{array}
\right.
\end{equation*}
We note that, for $y \in (y_\sigma, \infty)$,
$\sigma(y)<0$
and therefore $\widetilde{\sigma}(y)<0$ as well.
While this may appear odd,
it is not however an issue as Brownian increments are symmetric around the origin.
While Figure~\ref{pic:VIX_VVIX} provides strong empirical arguments in favour of such a model, a theoretical analysis thereof is however needed in order to investigate further its practical benefits.
For example, since $\widetilde{\sigma}(y)^2  \sim \gamma^2 y^{2(1-\beta)}$ as~$y$ approaches zero, 
the map $\widetilde{\sigma}$ is square integrable around the origin if and only if $\beta < \frac{3}{2}$,
and theoretical issues will arise if this is not satisfied (therefore ruling out such values out of calibration).
That said, as we will show below, our main interest will be on the behaviour of the process on $(y_\sigma,\infty)$,
and therefore this restriction on~$\beta$ will not be enforced.
Here and in the following,
given two functions~$f, g:\RR\to\RR$, we shall write
$f(y)\sim g(y)$ as~$y$ tends to a (possible infinite) point~$\bar{y}$ whenever
$\lim_{y\to \bar{y}} f(y)/g(y) = 1$.
We write $\mathbb{P}_y(\cdot)$ for the conditional probability $\mathbb{P}_y(\cdot| Y_0=y)$ and consequently $\mathbb{E}_y[\cdot]$ for $\mathbb{E}_y[\cdot| Y_0=y]$.

We now step into this theoretical analysis by first concentrating on the existence and uniqueness of a strong solution for~\eqref{eq:SDE}, then by deriving a precise classification of the special points~$0$, $y_\sigma$ and~$\infty$,
before diving into the asymptotic behaviour of the process
and the corresponding option prices.

\section{Main results}\label{sec:Results}

\subsection{Existence and uniqueness of strong solutions}\label{sec:existence_uniqueness_sol}

Following~\cite[Definition 2.1]{IW14},
the definition of a strong solution  allows for explosion in finite time.
It is enough to check existence and uniqueness of solutions for the one-dimensional equation associated to the process~$Y$
since the process $X= \log(S)$ is uniquely determined as a function of~$Y$ and~$W$.
A localised version of~\cite[Corollary to Theorem 3.2]{IW14} which, according to the authors, can be proved similarly to Theorem~3.1 therein, yields the existence of a unique strong solution, provided that the drift~$b(\cdot)$ is locally Lipschitz and the volatility~$\widetilde{\sigma}(\cdot)$ is $\half$-H\"older
(as mentioned in~\cite[page 184]{IW14}, this condition is in some sense maximal).
Alternatively, one can exploit~\cite[Proposition 2.3]{DW98}, since the conditions therein are a direct consequence of local Lipschitzianity and local H\"olderianity of~$b$ and~$\sigma$. Notice that the fact that we are just focusing here on the positive half-line and not on~$\mathbb{R}$ can be overcome just by setting~$\sigma$ and~$b$ identically equal to zero for negative arguments.
The local Lipschitz property of the drift is straightforward:
for any $N \in \mathbb{N}$ and any $x, y \in [-N,N]$, we have
\begin{align*}
    |b(x)-b(y)|
    & = \frac{\left|x-x^2 - \left(y-y^2\right)\right|}{h}
    \leq \frac{|x- y| + \left|x^2 - y^2\right|}{h}
    \leq \frac{|x- y|+2N|x - y|}{h}
    \leq K_N|x- y|,
\end{align*}
from which the local Lipschitz property with constant $K_N:=\frac{2N + 1}{h}$ follows.
Now, the volatility function is $\alpha$-H\"older with $\alpha \geq \half$ if and only if $0 < \beta \leq \half$ or $\beta=1$:
for $\beta=1$, $\widetilde{\sigma}(y)=-\frac{\alpha}{\beta}y+\gamma$ is affine hence globally Lipschitz.
Now, for any $N \in \mathbb{N}$ and $x, y \in [-N,N]$,
\begin{align*}
    \left|\widetilde{\sigma}(x)- \widetilde{\sigma}(y)\right|
    & = \left|-\frac{\alpha}{\beta}x +\gamma x^{1-\beta} + \frac{\alpha}{\beta}y - \gamma y^{1-\beta}\right|
    \leq \frac{\alpha}{\beta}|x-y| + \gamma \left|x^{1-\beta} - y^{1-\beta}\right|\\
    & \leq 2 \max \left\{\frac{\alpha}{\beta}, \gamma \right\}|x-y|^{1-\beta},
\end{align*}
which is locally
$\half$-H\"older continuous if and only if $\beta \leq \half$.
We will below consider the process~$Y$, not on the positive half line, but on the open interval $(y_{\sigma},\infty)$,
on which it enjoys nice ergodic properties.
There, the regularity of~$\widetilde{\sigma}$
is improved since, for any $x, y \in (y_\sigma, \infty)$,
\begin{align*}
   \left|\widetilde{\sigma}(x)- \widetilde{\sigma}(y)\right|
     & = \left|\frac{\alpha}{\beta}(y-x) +\gamma \left(x^{1-\beta}-y^{1-\beta}\right) \right|
    \leq \frac{\alpha}{\beta}|x-y| + \frac{\gamma (1-\beta)}{y_\sigma^{\beta}}|x - y|\\
     & \leq \left(\frac{\alpha}{\beta} +\gamma(1-\beta)\frac{\alpha}{\beta \gamma} \right)|x-y|
     = \frac{\alpha(2-\beta)}{\beta}  |x-y|.
\end{align*}
\subsection{Boundary classification}\label{sec:Boundary}
Now,
we need to analyse its behaviour in its domain and in particular at the boundary of the latter.
To do so, we follow the boundary classification in~\cite[Chapter 15, Section 6]{KT81}. 
The reason for this choice is that it seems the most suitable reference here. 
First, it includes both \textit{Feller} and \textit{Russian} boundary classifications, therefore allowing for a precise comparison. 
Second, it only requires the volatility coefficient to be non-null in the interior of the domain considered. 
On the contrary, the treatise in~\cite{CE05}, although more complete in some sense, requires the volatility process to be non null everywhere in~$\RR$. 
Consider a \textit{regular} (in the sense of~\cite{KT81}) diffusion process $Y=\{Y_t\}_{t \geq 0}$, on a domain $\Dom\subset\RR$, with left and right boundaries~$l$ and~$r$:
$$
\D Y_t= \mu(Y_t)\D t+ \sigma(Y_t)\D W_t, \qquad Y_0=y_0 \in \Dom.
$$
For any point~$y$ in the interior of~$\Dom$, namely $y \in (l,r)$, we assume that the drift and variance coefficients~$\mu(\cdot)$ and~$\sigma(\cdot)$ are continuous
and that $\sigma(y)>0$ for all~$y$ in the interior of~$\Dom$.
For any $x,y\in\RR$, introduce the hitting times $\tau_x := \inf \left\{ t \ge 0: Y_t = x \right\}$ and $\tau_{x,y} := \min \left\{\tau_x, \tau_y \right\}$.

We study the left boundary~$l$,
the discussion for the right boundary~$r$ being similar.
To provide a precise description, we recall 
some standard notions.
The \textit{scale function} $S:\Dom\to\RR$ is defined in terms of the so-called the \textit{scale density} $s:\Dom\to\RR$ via
\[
s(\xi):= \exp  \left( - \int_{\xi_0}^{\xi} \frac{2 \mu(v)}{\sigma^2(v)}\D v \right),
\qquad
S(x):=\int_{x_0}^x s(\xi) \D \xi, 
\]
where $\xi_0, x_0 \in (l,r)$ are arbitrary fixed points.
The particular choice of these points has no importance for the boundary discussion~\cite[Chapter 15, Section 3]{KT81}.
For any closed interval $I:=[a,b] \subset (l,r)$, 
we also introduce the \textit{scale measure}, namely the map 
$S:I\mapsto S(b) - S(a)$.
Then, we define the \textit{speed density}~$m$ and \textit{speed measure}~$M$:
\begin{align*}
&m(\xi) := \frac{1}{\sigma^2(\xi)s(\xi)}, 
\qquad
M[I] = M[a,b] := \int_a^b m(x) \D x.
\end{align*} 
Notice that both $S$ and $M$ are positive and finite on their domain.
Finally, 
\begin{align*}
N(l):=\int_l^x S[\eta,x] \D M(\eta)
	=\int_l^x M(l,\xi] \D S(\xi)
	= \int_l^x S[\eta,x]\frac{\D \eta}{\sigma^2(\eta)s(\eta)}.
\end{align*}
Since we are only interested in whether the integrals are finite or not, the upper bound~$x$ is irrelevant, 
explaining why we omit it from the notations.
Through the quantity $M(l, x]$, we can in some sense estimate the velocity of the process near $l$ and with the quantity $N(l)$ we can approximately quantify how long it takes to hit a point $x \in (l, r)$ starting at the left  boundary~$l$.
Now, we are ready to give the first classification:
\begin{definition}\ 
\begin{itemize}
\item[-] The left boundary~$l$ is \textit{attracting} if $S(l, x_0] < \infty$ for some $x_0$ in $(l, r)$. Then,
$$
\PP(\tau_{l_+}\leq \tau_{b}|Y_0=x)>0, \qquad \text{for any }l<x<b<r.
$$
\item[-] The left boundary~$l$ is \textit{unattracting} when $S(l, x_0] = \infty$ for some $x_0$ in $(l, r)$.
Then,
$$
\PP(\tau_{l_+} < \tau_{b}|Y_0=x)=0, \qquad \text{for any }l<x<b<r.
$$
\end{itemize}
\end{definition}
A left boundary~$l$ is therefore attracting when there is a positive probability that the process reaches~$l$  prior to the arbitrary state~$b$ (not necessarily in finite time), when its initial condition is $x < b$.
To complete our discussion of boundary classification we introduce the quantity $\Sigma(l):=\int_l^x S(l, \xi] \frac{\D\xi}{\sigma^2(\xi) s(\xi)}$,
where again the upper bound~$x$ in the integration is irrelevant.
Roughly speaking, $\Sigma(l)$ determines the time required by the process, starting from an interior point $x$, to reach the boundary~$l$ or another interior point $b>x$.
\begin{definition}
The boundary~$l$ is \textit{attainable} if $\Sigma(l) < \infty$; otherwise it is \textit{unattainable}.
\end{definition}
A straightforward argument shows that if~$l$ is attainable, then it is attracting. 
Indeed, $S(l, x_0] < \infty$ whenever $\Sigma(l) < \infty$. This is  in contrast to unattainable boundaries that may or may not be attracting.
For an attracting attainable boundary $l$, for any $l < x < b < r$,
\begin{align*}
	\PP(\tau_{l_+} < \infty| y_0=x)>0 
	\qquad \text{ and } \qquad
	\EE[\tau_{l_+} \wedge \tau_b | y_0=x] < \infty.
\end{align*}
Table 6.1 in~\cite{KT81} provides a complete portrait of Feller and Russian characterisations in terms of~$S(l,x]$, $M(l,x]$, $\Sigma(l)$ and~$N(l)$.
We give here a short description in words of Feller's:

\begin{itemize}

	\item[-]  \textbf{Regular boundary:} A regular boundary is attracting and attainable. 
	A diffusion process can enter but also leave from such a boundary point.

	\item[-]  \textbf{Exit boundary:} An exit boundary is attracting and attainable too, but when the initial point gets closer to it, the process cannot reach any interior point~$b$ regardless how close~$b$ is to~$l$. 
	Indeed, in this case it should hold: $\lim_{b \searrow l} \lim_{x \searrow l} \PP(\tau_b < t | Y_0=x)=0$, for any $t >0$. 
	No continuous sample path can exit $l$ after touching it.

	\item[-]  \textbf{Entrance boundary:} An entrance boundary is unattracting and unattainable. 
	A process starting from any point in the interior of the domain~$\Dom$ can not reach the entrance boundary. 
	Nevertheless, one can consider a process starting at the entrance boundary $l$: in this case, the process moves to the interior of the domain and never comes back to the boundary.

	\item[-]  \textbf{Natural (Feller) boundary:} A point is a natural boundary when it is unattainable (it can be attracting or not). 
	In general, such boundaries are discarded from the state space of the process since a diffusion process cannot start from nor reach it in finite time.
\end{itemize}
The following theorem, proved in Appendix~\ref{app:Proof_Thm_behaviour}, provides a detailed analysis of the behaviour of the process~$Y$ in Equation~\eqref{eq:SDE}
at the boundaries of its domain.
\begin{theorem}\label{Thm_behaviour}\ 
\begin{itemize}
\item[-]  Consider the process~$Y$
in~\eqref{eq:SDE} over the domain $\Dom=(y_\sigma, \infty)$. 
The right boundary $r=\infty$ is  entrance (unattracting, unattainable) while the left boundary $l=y_\sigma$ is
\begin{center}
	\begin{tabular}{|c|c|c|}
		\hline
		 Left boundary $y_\sigma$ & Feller & Russian \\ 
		\hline
		$y_\sigma > 1$ & exit-trap-absorbing  & attracting attainable \\  
		\hline
		$y_\sigma= 1$ & natural & attracting unattainable\\
		\hline
		 $y_\sigma < 1$ & entrance & unattracting unattainable\\
		\hline
	\end{tabular}
\end{center}
\item[-]  If the process~$Y$ in~\eqref{eq:SDE} is defined over $\Dom=(0,y_\sigma)$, then the left boundary $l=0$ is  
\begin{center}
	\begin{tabular}{|c|c|c|}
		\hline
		 Left boundary $0$ & Feller & Russian \\ 
		\hline
		$\beta<\half$ & regular  & attracting attainable \\  
		\hline
		$\beta\geq \half$ & exit-trap-absorbing & attracting attainable\\
		\hline
	\end{tabular}
\end{center}
while the right boundary $r=y_\sigma$ is
\begin{center}
	\begin{tabular}{|c|c|c|}
		\hline
		 Right boundary $y_\sigma$ & Feller & Russian \\ 
		\hline
		$y_\sigma > 1$ & entrance & unattracting unattainable\\
		\hline
		$y_\sigma= 1$ & natural & attracting unattainable\\
		\hline
		 $y_\sigma < 1$ & exit-trap-absorbing  & attracting attainable \\  
		\hline
	\end{tabular}
\end{center}
\end{itemize}
\end{theorem}

\begin{remark}
    As a consequence of this classification we limit our discussion to the domain $(y_\sigma, \infty)$ with $y_\sigma<1$. 
    On $(0,y_\sigma)$ the strict positivity of~$Y$ is not guaranteed as the origin is attracting and attainable. 
    Moreover, the case where $\Dom=(y_\sigma, \infty)$ with $y_\sigma  \geq 1$, should be ruled out as well since~$y_\sigma$ is attracting and attainable, so that~$Y$ may then exit it to enter $(0,y_\sigma)$ with strictly positive probability.
    An application of \cite[Theorem~3.2, Section~4, Chapter~4]{IW14} guarantees that~$Y$ does not explode in finite time, or more precisely that
$$
\PP\left(\inf\Big\{t \geq 0: Y_t \in \{y_\sigma,+\infty\}\Big\} = \infty \vert Y_0=y_0\right) = 1,
\qquad\text{for any }y_0 \in (y_\sigma, \infty).
$$
\end{remark}

\subsection{Ergodic Behaviour and Stationary distribution}\label{sec:Stationary}

\subsubsection{Ergodic  behaviour}\label{sec:ErgodicB}

We now discuss the ergodic behaviour of the process~$Y$ in~\eqref{eq:SDE} through the following theorem proved in Appendix~\ref{Proofthm:ErgodicB}.
To do so, introduce the probabilities
$$
\Pf(z) := \PP_{y_0}\left(\lim_{t \uparrow \tau_{\Dom}}Y_t=z\right),
\qquad\text{for }z \in \{l,r\},
$$
where~$\tau_{\Dom}$ denotes the lifetime of the process in $\Dom = (l,r)$.
Following~\cite[Section~2.7-2.8]{{Pinsky}}, transience of the process 
then corresponds to $\Pf(l)+\Pf(r) = 1$.
\begin{theorem}\label{thm:ErgodicB}
The ergodic behaviour of the process~$Y$ in~\eqref{eq:SDE} is as follows:
\begin{center}
	\begin{tabular}{|c|c|c|}
		\hline
		 & $\Dom=(0,y_\sigma)$ & $\Dom=(y_\sigma, \infty)$\\ 
		\hline
		$y_\sigma<1$ & $Y$ transient and  $\Pf(y_\sigma)$ in~\eqref{eq:Pysigma} & $Y$ recurrent \\ 
		\hline
		$y_\sigma=1$ & $Y$ transient and  $\Pf(y_\sigma)$ in~\eqref{eq:Pysigma}
		& $Y$ transient and $\Pf(y_\sigma)=1$\\
		\hline
		$y_\sigma > 1$ & $Y$ transient and  $\Pf(0)=1$   &  $Y$ transient and  $\Pf(y_\sigma)=1$   \\
		\hline
	\end{tabular}
\end{center}
with 
\begin{equation}\label{eq:Pysigma}
\Pf(y_\sigma)= \frac{\int_{0}^{y_0} \exp\left\{- \int_x^y \frac{2b(s)}{\widetilde{\sigma}^2(s)} \D s \right\} \D y}{\int_{0}^{y_\sigma} \exp\left\{- \int_x^y \frac{2b(s)}{\widetilde{\sigma}^2(s)} \D s \right\} \D y}
\quad \text{and} \quad
\Pf(0) = 1 - \Pf(y_\sigma),
\quad \text{for any }x \in (0,y_\sigma).
\end{equation}
\end{theorem}

An immediate consequence of the fact that~$Y$ is recurrent when $y_\sigma < 1$ on the domain $(y_\sigma, \infty)$  is that~$Y$ does not explode in finite time with probability one. 
This will thus be the case of interest, for which a stationary distribution is available
(Proposition~\ref{prop:Stationary}).

\subsubsection{Stationary distribution over the domain $\Dom=(y_\sigma, \infty)$}

We now investigate the ergodic properties of the process~$Y$ in Equation~\eqref{eq:SDE} over the domain $\Dom=(y_\sigma, \infty)$.
Recall that its infinitesimal generator is defined, for any $y \in \Dom$, as 
$$
(\Ll_{Y}\varphi)(y) := \lim_{t\downarrow 0}\frac{\EE[\varphi(Y_t) \vert Y_0=y] - \varphi(y)}{t},
$$
for all functions~$\varphi$ such that the limit is finite for all $y\in \Dom$.
We recall~\cite[Section 3.2]{FPSS11} that a process $(Y_t)_{t>0}$ is ergodic if it admits a unique, stationary distribution~$\Pi$, and for any measurable bounded function~$\phi$, the almost sure limit
$$
\lim_{t \uparrow \infty} \frac{1}{t} \int_0^t \phi(Y_s)\D s = \int_{\Dom} \phi(y) \Pi(\D y)
$$
holds.
If this limit exists, an ergodic solution must satisfy $\Ll_{Y}^* \Pi = 0$, 
where $\Ll_{Y}^*$ is the adjoint of the infinitesimal generator $\Ll_{Y}$, defined in~\cite[Section 1.5.3]{FPSS11} via the identity
\begin{equation}\label{def:adjoint}
\int g(\xi) \Ll_{Y} f(\xi) \D \xi = \int f(\xi) \Ll_{Y}^* g(\xi) \D\xi,
\end{equation}
for any rapidly decaying smooth test functions~$f$ and~$g$.
The generator and its adjoint are available explicitly here:
\begin{proposition}\label{pp:generator}
For any $y\in \Dom$, we have
\begin{align*}
(\Ll_Y f)(y) & = \frac{1}{h}y(1-y) \dy f(y) + \half y^2 \sigma^2(y) \dyy f(y),\\
(\Ll_Y^* g)(y) & = -\frac{1}{h} \dy \Big(y(1-y)g(y)\Big) + \half \dyy \Big(y^2\sigma^2(y)g(y)\Big).
\end{align*}
\end{proposition}

\begin{proof}
The expression for~$\Ll_Y$ is standard and the one for~$\Ll_{Y}^*$ follows using~\eqref{def:adjoint} and integration by parts. 
Given $y\in \Dom$ and $f,g: \Dom\to \RR$ twice continuously differentiable functions with bounded derivatives,
and such that the two functions and their derivatives tend to zero fast enough at the boundaries, we have
\begin{equation*}
\begin{aligned}
	\langle f, \Ll_Y^* g \rangle 
	& = \langle \Ll_Y f, g \rangle
	= \int_{\Dom}  \left[ \frac{y(1-y)}{h} 
	\dy f(y) + \half y^2 \sigma^2(y) \dyy f(y)\right]g(y) \D y\\
	& = \frac{1}{h}\int_{\Dom} \dy f(y) y (1-y)g(y) \D y
+ \half \int_{\Dom} \dyy f(y) y^2 \sigma^2(y) g(y) \D y\\
	& = \frac{1}{h}\int_{\Dom} \dy f(y) y (1-y)g(y) \D y
-\half \int_{\Dom} \dy f(y) \dy\left(y^2 \sigma^2 (y) g(y) \right)\D y \\
	& = -\int_{\Dom} \dy f(y) \left\{ -\frac{1}{h}y(1-y)g(y) +\half\dy \left(y^2\sigma^2(y)g(y) \right)\right\} \D y\\
	& = \int_{\Dom} f(y) \left[ -\dy \left(\frac{1}{h}y(1-y)g(y)\right) + \half \dyy\left(y^2\sigma^2(y)g(y)\right)\right] \D y,
\end{aligned}
\end{equation*}
and the proposition follows.
\end{proof}

For $f:\Dom\to\RR$, finding the explicit solution of the Poisson equation is tedious. 
Indeed,
$(\Ll_{Y}^* f)(y) = 0$
is equivalent to
$$
\begin{aligned}
\displaystyle
&\half y^2 f''(y) \left[ {\left(\frac{\alpha}{\beta}\right)}^2-\frac{2\alpha\gamma}{\beta} y^{-\beta}+\gamma^2 y^{-2\beta}\right] \\
&+yf'(y) \left[2\left\{ {\left(\frac{\alpha}{\beta}\right)}^2-\frac{\alpha\gamma}{\beta}(2-\beta)y^{-\beta}+(1-\beta)\gamma^2 y^{-2\beta}\right\} -\frac{1}{h}(1-y)\right] \\
&+f(y) \left[ {\left(\frac{\alpha}{\beta}\right)}^2 -\frac{\alpha\gamma}{\beta}(1-\beta)(2-\beta)y^{-\beta} + \gamma^2(1-\beta)(1-2\beta)y^{-2\beta} - \frac{1}{h}(1-2y) \right] = 0,
\end{aligned}
$$
with the constraint $\int_{\Dom} f(y) \D y = 1$.
This is a highly non-linear problem, 
which does not admit any obvious explicit solution.
However, using the probabilistic tools
developed in~\cite[page 242]{KT81}, 
such a closed-form expression can be derived as in the following proposition, proved in Appendix~\ref{prop:Stationary_proof}).

\begin{proposition}\label{prop:Stationary}
If $y_{\sigma}<1$, $\Dom=(y_\sigma, \infty)$, the unique stationary distribution reads
\begin{equation}\label{eq:StationaryY}
\Pi(\D y)= \left( \int_{y_\sigma}^{\infty}\frac{\D\xi}{\widetilde{\sigma}^2(\xi)s(\xi)}  \right)^{-1}\frac{\D y}{\widetilde{\sigma}^2(y) s(y)}.
\end{equation}
\end{proposition}


\subsection{Pricing PDE and expansion}\label{sec:PDE}
Pricing options on the stock price given in~\eqref{eq:SDE} can obviously be done with Monte Carlo simulations.
However, through Feynman-Kac, PDE techniques are (when available) often faster 
and may also (as we shall see below) provide closed-form expressions.
Consider an option with payoff~$h(X_T)$
at expiry~$T$,
and denote its price $P(t, X_t, Y_t)$ at time~$t\leq T$.
Introduce the operators 
\begin{equation}\label{eq:Operators}
\Ll_1:= y\sigma^2(y)\partial_{xy}
\qquad\text{and}\qquad
\Ll_{\BS}^{\sigma(y)}:= \partial_{t}+ \frac{\sigma^2(y)}{2}\dxx -\frac{\sigma^2(y)}{2}\dx.
\end{equation}
and recall that~$\Ll_Y$ is defined in Proposition~\ref{pp:generator},
while the operator~$\Ll_{\BS}^{\sigma(y)}$ is nothing else than the Black-Scholes infinitesimal generator 
with volatility~$\sigma(y)$.
\begin{proposition}\label{prop:PricingPDE}
Under the risk-neutral measure, the pricing PDE associated to~\eqref{eq:SDE} is
\begin{equation}\label{eq:PDE}
\left(\Ll_Y + \Ll_1+\Ll_{\BS}^{\sigma(y)}\right)P(t,x,y) = 0,
\end{equation}
for all $t\in [0,T)$, $x\in\RR$ and $y\in \Dom=(y_\sigma, \infty)$,
with terminal condition $P(T,x,y) = h(x)$.
\end{proposition}
Note that the PDE is stated in the domain $\Dom=(y_{\sigma}, \infty)$ and not on the whole positive half-line in the $y$-dimension. 
On~$\Dom$, the drift is quadratic (so that this representation follows from~\cite{Bishop} for example),
while the diffusion coefficient is at most of linear growth.
Note that since $\sigma(\cdot)$ is not bounded away from zero, 
the operator
$\Ll_{\BS}^{\sigma(y)}$ is not strictly elliptic, but only hypoelliptic.
Unfortunately, this pricing PDE does not admit an obvious explicit solution.
However, approximate solutions can be found by expanding the solution using perturbation methods, 
as developed in~\cite{FPSS11}. 
A key ingredient is the (unique) stationary distribution of the ergodic process~$Y$, 
which we proved above for the case $\Dom=(y_\sigma, \infty)$.
This perturbation analysis relies on a few other items
that we need to tackle. 
In particular, we assume that the pricing PDE admits a unique classical solution.

Following~\cite{FFF10, FFK12, FPS00, FPSS03, FPSS03B, FPSS11}, consider a `fast' version of the original process~$Y$, defined as
\begin{equation}
\label{eq:Model_XY_for_perturbation}
\left\{
\begin{array}{rll}
\D X_t & = \displaystyle -\half\sigma(Y_t)^2 \D t + \sigma(Y_t) \D W_t,  &X_0 = x_0\in\RR, \\
\D Y_t & = \displaystyle\frac{1}{\eps} b(Y_t) \D t + \frac{1}{\sqrt{\eps}} \widetilde{\sigma}(Y_t) \D W_t, &Y_0 = y_0 >0,
\end{array}
\right.
\end{equation}
for $\eps>0$.
Proposition~\ref{prop:PricingPDE} then implies that the option price~$P^\eps$, with payoff~$h$, satisfies
\begin{equation}\label{eq:pricing_eps}
\left[\frac{1}{\eps}\Ll_Y+\frac{1}{\sqrt{\eps}}\Ll_1+\Ll_{\BS}^{\sigma(y)}\right]P^\eps(t,x,y)=0,
\end{equation}
for all $t\in [0,T)$, $x\in\RR$ and $y\in \Dom$,
with boundary condition $P^{\eps}(T,x,y) = h(x)$.
Inspired by~\cite{FPS00, FPSS03, FPSS11}, 
we now provide an approximation for the price~$P^{\eps}$, proved in Appendix~\ref{prop:PerturPriceExpansion_Proof}:
\begin{proposition}\label{prop:PerturPriceExpansion}
If the payoff~$h$ is smooth, then the equality
$$
P^{\eps}(t,x,y) = P_{0}(t,x) + \sqrt{\eps}P_1(t,x) + \Oo(\eps)
$$
holds pointwise in $(t,x,y) \in [0,T)\times\RR\times \Dom$ as~$\eps$ tends to zero,
where~$P_0$ corresponds to the Black-Scholes price of the option having payoff $h$ with volatility
$\sigmaLim := \sqrt{\langle \sigma^2, \Pi\rangle}$
and
$$
P_1(t,x) = - \frac{T-t}{2}\langle \varpi, \Pi\rangle \left(\dx^3 - \dx^2 \right) P_0(t,x),
$$
for all $(x,t) \in \RR\times [0,T)$ with boundary condition $P_1(T,x)=0$
and with 
\begin{equation}\label{Eq_eq_def_varpi}
\varpi(y) := y\sigma^2(y)\psi'(y).
\end{equation}
Finally~$\psi$ is the unique solution to
\begin{equation}\label{Eq_eq_def_psi}
\Ll_Y \psi(y)= \sigma^2(y) - \sigmaLim^2,
\qquad\text{for all }y \in (y_{\sigma}, \infty).
\end{equation}
\end{proposition}
\begin{remark}
The assumption of a smooth payoff follows
that in~\cite{FPS00, FPSS11}.
Using mollification arguments, it could be relaxed to include standard European Call and Put options, but we leave this subtlety for later.
\end{remark}
\subsection{Small-time asymptotics}\label{sec:LDP}
We finally investigate the small-time behaviour of the solution to~\eqref{eq:SDE} using large deviations techniques, 
leading to closed-form asymptotics for option prices and implied volatilities.
We refer the reader to~\cite{FrizBook} for an overview of this topic.
For $\eps >0$, $t\in [0,T]$, introduce the small-time rescaling $(X^\eps_t, Y^\eps_t) :=(X_{\eps t} , Y_{\eps t})$, which satisfies
\begin{equation}
\label{eq:Model_X_eps}
\left\{
\begin{array}{rll}\displaystyle
\D X^\eps_t & = \displaystyle-\frac{\eps}{2} \sigma^2(Y^\eps_t) \D t + \sqrt{\eps}\, \sigma(Y^\eps_t) \D W_t,  &X^\eps_0 := x_0 \in\RR, \\
\D Y^\eps_t & = \displaystyle\eps b(Y^\eps_t) \D t + \sqrt{\eps}\,  \widetilde{\sigma}(Y^\eps_t) \D W_t, &Y^\eps_0 = y_0 >0.
\end{array}
\right.
\end{equation}

Let $\overline{\Hh}$ denote the space of absolutely continuous functions starting at the origin, 
with square integrable derivatives, such that
$$
\overline{\Hh}
:= \left\{
f: [0,T] \rightarrow \RR \text{ with } f = \int g(s) \D s \text{ for some } g\in L^2([0,T]), 
\text{ and }\inf_{t\in[0,T]} f_t \ge \frac{1}{\alpha} \log \left( 1 - y^\beta_0 \frac{\alpha}{\beta\gamma}\right)
\right\}.
$$

\begin{remark}
When $y_0 \ge y_\sigma$, the condition
\begin{equation}\label{eq:Condition_f}
\inf_{t\in[0,T]} f_t 
\ge \frac{1}{\alpha} \log \left( 1 - y^\beta_0 \frac{\alpha}{\beta\gamma}\right),
\end{equation}
is automatically satisfied and~$\overline{\Hh}$ is the usual Cameron-Martin space. 
When $y_0 < y_\sigma$, \eqref{eq:Condition_f} is needed to ensure that the solution of the controlled ODE introduced below is positive.
\end{remark}

We now state and prove (in Appendix~\ref{proof:LDP_X}) 
 a pathwise large deviations principle for the log-stock price process.
With $\mathrm{x}_0:= (x_0, y_0)$, 
introduce the map $\Ir^{X,Y}$ 
on~$\Cc([0,T], \RR\times\RR^*_+)$ by
$$
\Ir^{X,Y} (\mathrm{g})
:= \inf \left\{\Lambda (f), f\in \overline{\Hh}, \Ss^{\mathrm{x}_0} (f) = \mathrm{g} \right\},
$$
where~$\Lambda$ is the usual rate function driving the large deviations of the Brownian motion:
\begin{equation*}
\Lambda(f) :=
\left\{
\begin{array}{ll}
\displaystyle \half\int_{0}^{T}\left\|\dot{f}_t\right\|^2\D t, & \text{if } f \in \overline{\Hh},\\
\infty, & \text{otherwise},
\end{array}
\right.
\end{equation*}
and~$\Ss^{\mathrm{x}_0}(f)$ on~$[0,T]$ is the solution to the controlled ODE 
$\dot{\mathrm{g}}_t = \dot{f}_t \left(\underline{\sigma}(\mathrm{g}_t), \underline{\widetilde{\sigma}}(\mathrm{g}_t)\right)^\top$, with $\underline{{\sigma}}(x,y)=\sigma(y)$ and $\underline{\widetilde{\sigma}}(x,y)=\widetilde{\sigma}(y)$,
starting from $\mathrm{g}_0 = \mathrm{x}_0$.

\begin{theorem}\label{thm:LDP_X}
The rescaled $\log$-stock price process $X^\eps$ in~\eqref{eq:Model_X_eps} satisfies a pathwise large deviations principle on $\Cc ([0,T], \RR)$ as $\eps$ tends to zero with speed $\eps$ and rate function
\begin{equation}
\Ir^X (g) := \inf \Big\{ \Ir^{X,Y} (\mathrm{h}), \mathrm{h}:=(g,l), l \in\Cc([0,T], \RR^*_+), l_0 = y_0 \Big\},
\end{equation}
\end{theorem}
The proof of the theorem relies on first obtaining a large deviations principle for the rescaled process~$Y^\eps$, 
which we state below (and defer its proof to Appendix~\ref{proof:LDP_Y}).
Similarly to above, denote~$\Ss^y_2(f)$ the solution to the controlled ODE 
$\dot{g}_t = \widetilde{\sigma} (g_t) \dot{f}_t$, with $g_0 = y_0$.

\begin{proposition}\label{pp:LDP_Y}
The rescaled process $Y^\eps$ satisfies a pathwise large deviations principle on $\Cc([0,T],\RR^*_+)$ as $\eps$ tends to zero with speed $\eps$ and rate function
\begin{equation}
\label{eq:rf_LDP_ST}
\Ir^Y(g) :=
\inf \left\{ \Lambda (f), f\in \overline{\Hh}, \Ss_2^{y_0}(f) = g \right\}.
\end{equation}
\end{proposition}

Large deviations have been used extensively in Mathematical Finance to derive asymptotic behaviours of the implied volatility and we refer the reader to the monograph~\cite{FrizBook} for a thorough overview. 
The latter, $\Sigma_t (k)$, is the unique non-negative solution to
$\Cc_{\BS}(t, \E^k, \Sigma_t (k)) = \Cc_{obs}(t, \E^k)$, with $\Cc_{obs}(t, \E^k)$ an observed (or computed) Call option price with maturity~$t$ and strike~$\E^k$,
and~$\Cc_{\BS}$ is the corresponding Call price in the Black-Scholes model.
A large deviations principle is the first step to understand the short-time behaviour of the process,
and going from there to the corresponding behaviour of the implied volatility requires
a few small steps that we follow in Appendix~\ref{cor:SmallTimeIV_proof}.

\begin{corollary}\label{cor:SmallTimeIV}
For $y_0 \in \Dom =(y_\sigma, \infty)$, with $y_\sigma < 1$, small-time out-of-the-money options behave as
\begin{equation*}
\begin{array}{ll}
\lim_{t \downarrow 0} t \log \EE\left[\left(S_t - \E^k\right)_{+}\right]
= - \inf _{y \geq k} \mathrm{I}^X (g)\rvert_{g(1)= y},
& \text{if }k>0,\\
\lim_{t \downarrow 0} t \log \EE\left[\left(\E^k - S_t\right)_{+}\right]
=  - \inf _{y \leq k} \mathrm{I}^X (g)\rvert_{g(1)= y},
& \text{if }k<0.
\end{array}
\end{equation*}
\end{corollary}

Similarly to~\cite[Theorem 2.4]{FJ09Heston},
we finally deduce the behaviour of the short-time smile:
\begin{corollary}\label{cor:asymptoticIV}
For $y_0 \in \Dom =(y_\sigma, \infty)$, with $y_\sigma < 1$, the implied volatility behaves as
\begin{equation*}
\lim_{t \downarrow 0} \Sigma_t (k) = 
\left\{\begin{array}{ll}
\displaystyle \frac{k^2}{2} \left(\inf _{y \geq k} \mathrm{I}^X (g)\rvert_{g(1)= y} \right)^{-1},
& \text{if }k>0,\\
\displaystyle \frac{k^2}{2} \left(\inf _{y \leq k} \mathrm{I}^X (g)\rvert_{g(1)= y} \right)^{-1},
& \text{if }k<0.
\end{array}
\right.
\end{equation*}
\end{corollary}

\appendix

\section{Proof of Theorem~\ref{Thm_behaviour}}
\label{app:Proof_Thm_behaviour}

The proof below relies on the techniques developed in~\cite[Chapter 15, Section 6]{KT81}. 

\subsection{Proof for the domain $\Dom=(y_\sigma, \infty)$}		\subsubsection{Left boundary $y_\sigma$}	
	The classification of the left boundary $y_\sigma$ follows from Lemma~\ref{lem_y_sigma_KT}.
	Introduce, on the domain $(0, \infty)$, the process $Z := (Y - y_\sigma)$, satisfying the SDE
	$$
	\D Z_t = \overline{b}(Z_t) \D t + \overline{\sigma}(Z_t) \D W_t, 
	\qquad Z_0 := y_0 - y_\sigma >0,
	$$
	with $\overline{b}(z) := b(z+y_\sigma)$ 
	and $\overline{\sigma}(z) := \widetilde{\sigma}(z+y_\sigma)$, for $z>0$.
	Armed with Lemma~\ref{lem_y_sigma_KT},  we attack the boundary classification of the origin for~$Z$, 
	which corresponds to the classification of the left boundary~$y_\sigma$ for the original process~$Y$.
	The classification for the different cases, namely $y_	\sigma > 1$, $y_\sigma = 1$ and $y_\sigma < 1$ follows from a 	careful inspection of~\cite[Table 6.1, Chapter 15]{KT81} together with~\cite[Lemma 6.3, Chapter 15]{KT81}.
	Introduce
\begin{equation}\label{eq_def_s_bar}
\begin{array}{rlrlrl}
    \displaystyle \overline{s}(x)
    & := \displaystyle \exp\left\{ \int_x^{\overline{a}} \frac{2\overline{b}(y)}{\overline{\sigma}^2(y)} \D y \right\},
     & 
    \displaystyle \overline{S}(0,{\overline{a}}]
    & := \displaystyle \int_0^{\overline{a}} \overline{s} (y) \D y,
     & 
    \displaystyle \overline{S}[x,{\overline{a}}]
     & := \displaystyle \int_x^{\overline{a}} \overline{s} (y) \D y,\\
        \displaystyle \overline{M}(0,{\overline{a}}]
    & := \displaystyle \int_0^{\overline{a}} \frac{\D x}{\overline{\sigma}^2(x)\overline{s}(x)},
     & 
    \displaystyle \overline{\Sigma}(0)
    & := \displaystyle \int_0^{\overline{a}} \frac{\overline{S}(0,x]}{\overline{\sigma}^2(x)\overline{s}(x)} \D x,
    & \displaystyle \overline{N}(0) & := \displaystyle \int_0^{\overline{a}} \frac{\overline{S}[x,{\overline{a}}]\D x}{\overline{\sigma}^2(x)\overline{s}(x)},
\end{array}
\end{equation}
for $\overline{a} > 0$ and $x\in (0,\overline{a}]$. 
We then deduce their behaviour.
\begin{lemma}\label{lem_y_sigma_KT}
The following hold:
\begin{equation*}
\begin{array}{llll}
\displaystyle  \overline{S}(0,{\overline{a}}]< \infty,
&\displaystyle  \overline{M}(0,{\overline{a}}] = \infty,
&\displaystyle  \overline{\Sigma}(0)< \infty,
&\displaystyle  \text{if }y_\sigma > 1,\\
\displaystyle  \overline{S}(0,{\overline{a}}]< \infty,
&
&\displaystyle  \overline{\Sigma}(0) = \infty,
&\displaystyle  \text{if }y_\sigma = 1,\\
\displaystyle  \overline{S}(0,{\overline{a}}] = \infty,
& 
&\displaystyle  \overline{N}(0) < \infty,
&\displaystyle  \text{if }y_\sigma < 1.
\end{array}
\end{equation*}
\end{lemma}

\begin{remark}
We do not need $\overline{M}(0,{\overline{a}}]$ in the second and third cases because only a few combinations for boundedness/unboundedness of these quantities are possible. 
They are displayed in~\cite[Table 6.1, page 233]{KT81}. 
In particular, in the second line, $\overline{S}(0,{\overline{a}}]$ can be finite and $\overline{\Sigma}(0)$ not if and only if $\overline{M}(0,{\overline{a}}]$ and $\overline{N}(0)$ are infinite. 
In the third line $\overline{N}(0)$ finite implies $\overline{M}(0,{\overline{a}}]$ finite and $\overline{S}(0,{\overline{a}}]$ infinite implies $\overline{\Sigma}(0)$ infinite by~\cite[Lemma 6.3, page 231]{KT81}. 
Similar arguments motivate the form of the statements in Lemmas~\ref{Lem_bb_r_infty}-\ref{Lem_bb_l_zero}-\ref{lemma_y_sig_2}.
\end{remark}

\begin{proof}[Proof of Lemma~\ref{lem_y_sigma_KT}]
We start with the limiting behaviour of the function~$\overline{s}$ and its integral, the scale measure.
Since
\[
\frac{2\overline{b}(y)}{\overline{\sigma}^2(y)}
=\frac{2(1-y_\sigma-y)}{h\gamma^2(y+y_\sigma)^{1-2\beta}\left(1-(1+\frac{y}{y_\sigma})^\beta\right)^2},
\]
a straightforward Taylor expansion around the origin yields
\begin{equation}\label{eq:Expan_right_m}
\left(1-\left(1+\frac{y}{y_\sigma}\right)^\beta\right)^{-2}
= \frac{y^2_\sigma}{\beta^2 y^2}
\left[ 1 - \chi_1 \frac{y}{y_\sigma} + \chi_2 
\frac{y^2}{y_\sigma^2}
-\chi_3\frac{y^3}{y_\sigma^3} + \Oo(y^4) 
\right],
\end{equation}
with 
\begin{equation}\label{eq:ABCbeta_1}
\chi_1 :=  {\beta-1},
\qquad
\chi_2 := \frac{(5\beta-1)(\beta-1)}{12},
\qquad
\chi_3 := \displaystyle  \frac{(\beta^2-1)\beta}{12}.
\end{equation}

Introduce
$\tbrK := \frac{2y^{2\beta +1}_\sigma (1- y_\sigma)}{h \beta^2 \gamma^2}$ 
, and 
$\tbrKab := -\frac{1}{{\overline{a}}} + \tbrchif \log ({\overline{a}}) + \tbrchis{\overline{a}}$.
Using~\eqref{eq:Expan_right_m}, we obtain the asymptotic behaviour, as~$y$ approaches zero,
\begin{equation}\label{eq:Expan_right_integrand}
    \frac{2\overline{b}(y)}{\overline{\sigma}^2(y)}
    = \frac{\tbrK}{y^2} \left( 1 + \tbrchif y + \tbrchis y^2 -\tbrchit y^3 + \Oo\left(y^4\right) \right),
\end{equation}
for some constants $\tbrchif, \tbrchis, \tbrchit$ depending on $\chi_1, \chi_2, \chi_3$,
the values of which are not important\footnote{In the following the symbols $\tbrchif, \tbrchis, \tbrchit$ will refer to different quantities whose specific values are not important for the convergence of the quantities we are interested in.}.
Notice that $\tbrK < 0$, as $y_\sigma > 1$.
Since the expansion is uniform on $[x,{\overline{a}}]$, one obtains
\begin{align}\label{eq:Expan_right_rho_m}
    \overline{s} (x) 
    &= \exp\left\{ 2\int_x^{\overline{a}} \frac{\overline{b}(y)}{\overline{\sigma}^2 (y)} \D y \right\}
    = \exp\left\{\frac{\tbrK}{x}\left(1-\overline{\chi_1} x \log x\right)\right\} 
    \E^{\tbrK \tbrKab} \exp\left\{- \overline{\chi_2} \tbrK x + \Oo(x^2)\right\}\nonumber\\
    &=\exp\left\{\frac{\tbrK}{x} + \tbrK \tbrKab\right\} x^{-\tbrK\overline{\chi_1}} (1+\Oo(x)).
\end{align}
Since $\tbrK < 0$, $\overline{s}(x)$ tends to zero as~$x$ tends to zero from above, and $\overline{S}(0,{\overline{a}}]$ is finite. 

In the case $y_\sigma < 1$, the expansion~\eqref{eq:Expan_right_rho_m} is still valid, albeit with $\tbrK > 0$. 
Therefore,~$\overline{s}$ explodes at the origin and $\overline{S}(0,{\overline{a}}]=\int_0^{\overline{a}} \overline{s} (y) \D y$ is infinite.

The case $y_\sigma = 1$ is slightly different and has to be studied separately. First of all, a Taylor expansion around the origin provides
\begin{equation}\label{eq:Expan_right_1_m}
\left[1- \left(1+y\right)^{\beta} \right]^{-2}
= \frac{1}{\beta^2 y^2} \left[ 1 - \chi_1 y + \chi_2 
y^2 -\chi_3 y^3 + \Oo\left(y^4\right) \right],
\end{equation}
with $\chi_1, \chi_2$ and $\chi_3$ in~\eqref{eq:ABCbeta_1}.
This implies, as~$y$ approaches zero, 
\begin{equation}
\label{eq:Expan_right_integrand_1_m}
\frac{2\overline{b}(y)}{\overline{\sigma}^2(y)}
= - \frac{2y}{h \gamma^2 (y+1)^{1-2\beta}} \frac{1}{(1-(1+y)^\beta)^2}
= \tbrooK \left( \frac{1}{y} + \tbrchif  + \tbrchis x + \tbrchit y^2+ \Oo\left(y^3\right) \right),
\end{equation}
with 
$\tbrooK := -\frac{2}{h \beta^2 \gamma^2} < 0$, and 
$\tbrooKab := \log({\overline{a}}) + \tbrchif {\overline{a}} + \frac{\tbrchis}{2} {\overline{a}}^2 +\frac{\tbrchit}{3}{\overline{a}}^3$. 
Since the expansion is uniform on $[x,{\overline{a}}]$, one obtains
\begin{equation}\label{eq:Expan_right_rho_1_m}
\begin{aligned}\displaystyle
    \overline{s} (x) 
    &= \exp\left\{ 2\int_x^{\overline{a}} \frac{\overline{b}(y)}{\overline{\sigma}^2 (y)} \D y \right\}
    = \exp\left\{\tbrooK\left(-\log(x)-{\tbrchif}x + \Oo(x^2)\right)\right\}
    \exp\left\{  \tbrooK\, \tbrooKab \right\}\\
    &=\exp\left\{-\tbrooK \log(x) + \tbrooK\, \tbrooKab\right\} (1+\Oo(x))
    =\exp\left\{ \tbrooK\, \tbrooKab\right\} x^{-\tbrooK}(1+\Oo(x)).
\end{aligned}
\end{equation}
Since $\tbrooK<0$, $\overline{s}(x)$ tends to zero as~$x$ tends to zero
and therefore $\overline{S}(0,{\overline{a}}]=\int_0^{\overline{a}} \overline{s} (x) \D x$ is finite. 

Now, for  $y_\sigma > 1$, it is straightforward to see that 
\begin{displaymath}
    \overline{M}(0,{\overline{a}}]
    =\int_0^{\overline{a}} \frac{\D x}{\overline{s}(x)\overline{\sigma}^2(x)}  
    = \int_0^{\overline{a}} \frac{\D x}{\overline{s}(x)\widetilde{\sigma}^2(x+y_\sigma)}
    \ge \int_0^{\overline{a}} \frac{\D x}{\widetilde{\sigma}^2(x+y_\sigma)},
\end{displaymath}
which is clearly infinite, because~$\overline{s}$ is bounded above by~$1$ on~$(0,{\overline{a}}]$ and from the asymptotic behaviour of the integrand around the origin. Indeed, for $\eps >0$, 
\begin{align}\label{eq:Isolated}
\int_{0}^{\eps} \frac{1}{\widetilde{\sigma}^2(x+y_\sigma)} \D x
&
\ge \frac{\tbrepsK}{\gamma^2} \int_0^{\eps} \left( 1- \left( 1+\frac{x}{y_\sigma} \right)^\beta \right)^{-2}\D x
\end{align}
with $\tbrepsK := (y_\sigma + \eps)^{-2(1-\beta)}$ for $\beta\in (0,1)$, $\tbrepsK := y_\sigma^{2(\beta-1)}$ for $\beta >1$ and $\tbrepsK := 1$ for $\beta =1$.
Recalling the Taylor expansion in~\eqref{eq:Expan_right_m}, we see that the integrand is not integrable around zero.
We thus conclude about the right behaviour of $y_\sigma$ by noting that the integral representation of $\overline{M}(0,{\overline{a}}]$ diverges.

We now prove the last statement of the lemma, and start with the case $y_\sigma > 1$.
Using~\eqref{eq:Expan_right_rho_m}, we write the asymptotic behaviour of~$\overline{S}(0,x]$ around the origin 
by integrating the asymptotic behaviour of $\overline{s}(\cdot)$ around zero. 
Classical asymptotic expansions for integrals~\cite[Chapter~3.3, pages~62 and~67]{Miller} 
(note that the leading contribution arises at the right boundary of the integration domain) yields, 
after the change of variable $y \mapsto z x$,
\begin{align*}
    \overline{S}(0,x]
    & =  \int_0^x \overline{s} (y) \D y\\
    & =  \E^{\tbrK \tbrKab} x^{-\tbrK\tbrchif +1}\int_0^1 \exp\left\{ \frac{\tbrK}{zx} \right\} z^{-\tbrK\tbrchif} \D z
    =  \E^{\tbrK \tbrKab} x^{-\tbrK\tbrchif+1} \exp\left\{ \frac{\tbrK}{x} \right\} \left( -\frac{x}{\tbrK} + \Oo(x^2) \right),
\end{align*}
as $x$ tends to zero. Combining this with~\eqref{eq:Expan_right_m} and~\eqref{eq:Expan_right_rho_m}, we obtain
\begin{align*}
    \frac{\overline{S}(0,x]}{\overline{s}(x) \overline{\sigma}^2 (x)}
    = \frac{y^{2\beta}_\sigma}{\beta^2 \gamma^2} x \left(-\frac{x}{\tbrK} + \Oo(x^2) \right) \frac{1}{x^2} \left( 1 + \Oo(x) \right) 
    = -\frac{y^{2\beta}_\sigma}{\tbrK \beta^2 \gamma^2} \left( 1 + \Oo(x) \right),
\end{align*}
which is integrable on $(0,{\overline{a}}]$ and concludes the proof, using the fact that
$$
\overline{\Sigma}(0)
= \int_0^{\overline{a}} \frac{\overline{S}(0,x]}{\overline{s}(x) \overline{\sigma}^2(x)} \D x < \infty.
$$
When $y_\sigma = 1$, using~\eqref{eq:Expan_right_rho_1_m}, we can write the asymptotic behaviour of~$\overline{S}(0,\cdot]$ around the origin by integrating that of~$\overline{s}(\cdot)$ around zero. 
This yields, after the change of variable $y \mapsto z x$,
$$
\overline{S}(0,x]
=  \int_0^x \overline{s} (y) \D y
=  \E^{\tbrooK \tbrooKab} x^{-\tbrooK+1}
 \int_0^1 z^{-\tbrooK} \D z
=  \frac{\E^{\tbrooK \tbrooKab}}{1-\tbrooK} x^{-\tbrooK+1}
\left( 1 + \Oo(x) \right),
$$
as $x$ tends to zero. 
Exploiting this together with~\eqref{eq:Expan_right_1_m} and~\eqref{eq:Expan_right_rho_1_m}, we obtain
\begin{align*}
    \frac{\overline{S}(0,x]}{\overline{s} (x) \overline{\sigma}^2 (x)}
    = \frac{1}{\beta^2 \gamma^2 (1-\tbrooK)}\frac{1}{x}
    \left( 1 +\Oo(x) \right),
\end{align*}
which is not integrable on $(0,{\overline{a}}]$ and thus
$$
\overline{\Sigma}(0)
=\int_0^{\overline{a}} \frac{\overline{S}(0,x]}{\overline{s} (x) \overline{\sigma}^2(x)} \D x = \infty.
$$

When $y_\sigma < 1$, we look at $\overline{N}(0)=\int_0^{\overline{a}} \frac{\overline{S}[x,{\overline{a}}]}{\overline{\sigma}^2(x)\overline{s}(x)} \D x$.
Since both $\overline{S}[a,\cdot]$ and $\overline{s}(\cdot) \overline{\sigma}^2(\cdot)$ diverge to infinity at the origin, 
we study the behaviour of the integrand around zero.
For $\delta >0$ such that $x < \delta < {\overline{a}}$ and $x>0$,
$
\overline{S}[x,{\overline{a}}]
=\int_x^{\overline{a}} \overline{s} (y) \D y 
= \int_x^{\delta} \overline{s} (y) \D y + \int_{\delta}^{\overline{a}} \overline{s} (y) \D y$. 
The second integral exists since~$\overline{s}$ is continuous on compacts in~$\RR_+$.
Regarding the first one, classical asymptotic expansions for integrals and~\eqref{eq:Expan_right_rho_m}, yield, after the change of variable $y \mapsto zx$,
$$
\int_x^{\delta} \overline{s} (y) \D y
= \E^{\tbrK \tbrKab} x^{1-\tbrK\tbrchif} \int_1^{\delta / x} \exp\left\{ \frac{\tbrK}{xz} \right\} z^{-\tbrK\tbrchif} \D z
= \E^{\tbrK \tbrKab} x^{1-\tbrK\tbrchif} \E^{\frac{\tbrK}{x}} \left(\frac{x}{\tbrK} + \Oo(x^2)\right),
$$
and the asymptotic behaviour of the integrand around the origin becomes
$$
\frac{\overline{S}[x,{\overline{a}}]}{\overline{s} (x) \overline{\sigma}^2(x)}
= \frac{y^{2\beta}_\sigma}{\beta^2 \gamma^2} \left(\frac{x^2}{\tbrK} + \Oo(x^3)\right) \frac{1}{x^2} (1+\Oo(x))
= \frac{y^{2\beta}_\sigma}{\beta^2 \gamma^2 \tbrK} (1+ \Oo(x)),
$$
which is integrable at the origin, and the claim is proved.
\end{proof}

\subsubsection{Right boundary $\infty$}
To end the boundary classification for the first domain we are left to study the behaviour at the right boundary $\infty$.
We exploit the Lemma~\ref{Lem_bb_r_infty} and~\cite[Table 6.1, Chapter 15]{KT81}.
First, for $y > y_\sigma$, let $s(y)=\exp\left\{ - \int_a^{y}\frac{2b(x)}{\widetilde{\sigma}^2(x)}\D x \right\}$, with $a > y_\sigma$ fixed.
Then, recall the definitions of some quantities:
\begin{equation*}
		\begin{array}{lll}
		\displaystyle S[a,\infty)=\int_a^{\infty} s(x)\D x,
		&\displaystyle   M[a,\infty)=\int_a^\infty \frac{\D x}{\widetilde {\sigma}^2(x) s(x)},
		&\displaystyle  N(\infty)=\int_a^{\infty} \frac{S[a,x]}{\widetilde {\sigma}^2(x) s(x)} \D x.
		\end{array}
\end{equation*}
\begin{lemma}\label{Lem_bb_r_infty}
The following hold:
$$
S[a,\infty)=\infty,\qquad\qquad
M[a,\infty) < \infty,\qquad\qquad
N(\infty) < \infty.
$$
\end{lemma}

\begin{proof}
As $y$ tends to infinity,
$$
\frac{2b(y)}{\widetilde{\sigma}^2(y)} = \frac{2(1-y)}{hs}\left(-\frac{\alpha}{\beta}+\gamma y^{-\beta}\right)^{-2} \sim -\frac{2}{h}\left(\frac{\beta}{\alpha}\right)^2,
$$
and therefore, as $x$ tends to infinity,
\begin{equation}\label{Eq_bb_inf_int_s}
-\int_{a}^x\frac{2b(y)}{\widetilde{\sigma}^2(y)} \D y
\sim \frac{2}{h}\left(\frac{\beta}{\alpha}\right)^2 x.
\end{equation}
Then, $ S[a,\infty)=\int_a^{\infty}
\exp \left( - \int_a^{x}\frac{2b(y)}{\widetilde{\sigma}^2(y)}\D y \right)\D x \sim \int_a^{\infty} \exp\left(\frac{2}{h}\left(\frac{\beta}{\alpha}\right)^2 x\right) \D x$ is infinite.
Now $M[a, \infty)$ is finite since 
$\frac{1}{\widetilde {\sigma}^2(x) s(x)}\sim \frac{\exp\left(-\frac{2}{h}\left(\frac{\beta}{\alpha}\right)^2 x\right)}{x^2\left(-\frac{\beta}{\alpha}+\gamma x^{-\beta}\right)}$
as $x \uparrow \infty$, which is integrable.
Finally, $S[a,x]\sim \left[\frac{2}{h}\left(\frac{\beta}{\alpha}\right)^2\right]^{-1}\exp\left(\frac{2}{h}(\frac{\beta}{\alpha})^2 x\right)$
    as $x \uparrow \infty$, and so $\frac{ S[a,x]}{\widetilde {\sigma}^2(x)s(x)} \sim  \frac{\alpha^3 h}{2 \beta^3}\frac{1}{x^2}$, which is integrable at infinity and $N(\infty)$ is finite.
\end{proof}

\subsection{Proof for the domain $\Dom=(0,y_\sigma)$}

\subsubsection{Left boundary $0$}
Consider $\widetilde{s}(x)=\exp\left\{\int_x^{\widetilde{a}} \frac{2b(\xi)}{\widetilde{\sigma}^2(\xi)}\D \xi\right\}$ 
for $y_\sigma > \widetilde{a} > 0$, $x\in (0,\widetilde{a}]$, and
$$
\widetilde{S}(0,\widetilde{a}]
:=\int_0^{\widetilde{a}}  \widetilde s (y) \D y,\qquad\qquad
\widetilde{M}(0,\widetilde{a}]
    :=\int_0^{\widetilde{a}} \frac{\D x}{\widetilde{\sigma}^2(x)\widetilde s(x)} ,\qquad\qquad
\widetilde{\Sigma}(0)
    :=\int_0^{\widetilde{a}} \frac{\widetilde  S(0,x]}{\widetilde {\sigma}^2(x)s(x)}  \D x.
$$
The following lemma, together with~\cite[Table 6.1, Chapter 15]{KT81}, helps for the left boundary~$0$:

\begin{lemma}\label{Lem_bb_l_zero}
For any $\widetilde{a} \in (0, y_\sigma)$,
\begin{equation*}
\begin{array}{llll}
    \displaystyle   \widetilde{S}(0,\widetilde{a}]
    < \infty,
    &\displaystyle  \widetilde{M}(0,\widetilde{a}] < \infty,
    &\displaystyle  \widetilde  \Sigma(0) < \infty,
    &\displaystyle  \quad \text{ if }\beta < \half,\\
    \displaystyle  \widetilde{S}(0,{\widetilde{a}}] < \infty,
    &\displaystyle \widetilde{M}(0,{\widetilde{a}}] = \infty,
    &\displaystyle  \widetilde{\Sigma}(0) < \infty,
    &\displaystyle  \quad \text{ if } \beta \geq \half.
\end{array}
\end{equation*}
\end{lemma}

\begin{proof}[Proof of Lemma~\ref{Lem_bb_l_zero}]
We start by showing that $\widetilde{S}(0,{\widetilde{a}}]$ is always finite.
The only possible issue for integrability is at zero, so we expand the integrand in a neighborhood of the origin:
\begin{equation*}
\frac{2b(s)}{\widetilde{\sigma}^2(s)} 
= \frac{2(1-s)}{h\gamma^2 s^{1-2\beta}}\left[ 1- \frac{\alpha}{\beta \gamma} s^{\beta} \right]^{-2}
= \frac{2 s^{2\beta-1}}{h\gamma^2}
\left[ 1 + \frac{2\alpha}{\beta \gamma} s^{\beta} +3 \left(\frac{\alpha}{\beta \gamma}\right)^2 s^{2\beta}  + \Oo(s^{3\beta}) \right].
\end{equation*}
Since the expansion is uniform on $(0,\widetilde{a})$,
\begin{equation*}
-\int_{\widetilde{a}}^{x} \frac{2b(s)}{\widetilde{\sigma}^2(s)} \D s
= -\frac{2}{\beta h \gamma^2}\left(\frac{x^{2\beta}}{2} + \frac{2\alpha}{3\beta \gamma}x^{3 \beta} + \frac{3\alpha^2}{4\beta^2\gamma^2}x^{4\beta} + \Oo(x^{5 \beta})+ \tbrooKab\right),
\end{equation*}
with $\tbrooKab:= -\frac{{\widetilde{a}}^{2\beta}}{2} - \frac{2\alpha}{3\beta \gamma}{\widetilde{a}}^{3 \beta} - \frac{3\alpha^2}{4\beta^2\gamma^2}{\widetilde{a}}^{4 \beta}$,
we obtain
\begin{equation*}
    \widetilde{s}(x)
    = \exp\left\{-\int_{{\widetilde{a}}}^{x} \frac{2b(s)}{\widetilde{\sigma}^2(s)} \D s\right\}
    = \exp\left\{-\frac{x^{2\beta} + \Oo(x^{3 \beta})}{\beta h \gamma^2}
    - \frac{2 \tbrooKab}{\beta h \gamma^2}\right\}
    = \exp\left\{ -\frac{2}{\beta h \gamma^2}\tbrooKab\right\}\left(1+\Oo\left(x^{2\beta}\right)\right),
\end{equation*}
and so $\widetilde{S}(0,\widetilde{a}]$ is always finite.

Now, around zero we have the Taylor expansions
\begin{equation}\label{eq_T_ex_ss}
\begin{split}
    &\widetilde s(x)=\exp\left\{ - \frac{2\tbrooKab}{\beta h \gamma^2}\right\}\left(1+\Oo\left(x^{2\beta}\right)\right),\\
    &\widetilde{\sigma}^2(x)
    =\frac{1}{\gamma^2x^{2-2\beta}}\left(1-\frac{\alpha}{\beta \gamma}x^\beta\right)^2
    =\frac{1}{\gamma^2x^{2-2\beta}}\left(1+2\frac{\alpha}{\beta \gamma}x^\beta+ 3\frac{\alpha^2}{\beta^2 \gamma^2}x^{2\beta} + \Oo(x^{3\beta})\right)
\end{split}
\end{equation}
and so $(\widetilde{\sigma}^2(x)\widetilde s(x))^{-1}=\exp\left\{ - \frac{2}{\beta h\gamma^2}\tbrooKab\right\}\frac{1}{\gamma^2x^{2-2\beta}}\left(1+2\frac{\alpha}{\beta \gamma}x^\beta + \Oo(x^{2\beta})\right)$, which is integrable around zero, and so $\widetilde  M(0,{\widetilde{a}}]$ finite, if and only if $\beta >\half$.

Finally, we easily compute a Taylor expansion for $\widetilde{S}(0,x]$ around the origin by integration:
\begin{equation}\label{eq_T_ex_sss}
\widetilde{S}(0,x]=\exp\left\{ - \frac{2}{\beta h \gamma^2}\tbrooKab\right\}x\left(1+\Oo\left(x^{2\beta}\right)\right),
\end{equation}
hence using~\eqref{eq_T_ex_ss}, 
the Taylor expansion around the origin for the integrand in~$\widetilde{\Sigma}(0)$ reads
\[
\frac{\widetilde{S}(0,x]}{\widetilde{\sigma}^2(x)\widetilde s(x)}
 = \frac{1+\Oo(x^\beta)}{\gamma^2x^{1-2\beta}}.
\]
Since this is integrable around the origin if and only if $\beta >0$, $\widetilde{\Sigma}(0)$ is finite for all $\beta>0$.
\end{proof}

\subsubsection{Right boundary $y_\sigma$}

The strategy to prove the third table in the theorem is similar, albeit with different computations, to the first case.
Similarly to before, introduce the process $\widehat{Z} := y_\sigma - Y$ satisfying the SDE
$$
\D \widehat{Z}_t = \widehat{b}(Z_t) \D t + \widehat{\sigma}(\widehat{Z}_t) \D W_t, 
\qquad \widehat{Z}_0 := y_\sigma - y_0 >0,
$$
as well as the maps $\widehat{b}(x) := -b(y_\sigma - x)$ and $\widehat{\sigma}(x) := -\widetilde{\sigma}(y_\sigma -x)$ for $x>0$.
With Lemma~\ref{lemma_y_sig_2}, we obtain the boundary classification of the origin as left boundary for~$\widehat{Z}$ on the domain $\Dom=(0, y_\sigma)$ corresponding to the right boundary classification of~$y_\sigma$ for~$Y$ on the same domain.
All the cases $y_\sigma <1$, $y_\sigma = 1$ and  $y_\sigma > 1$ follow~\cite[Table 6.1, Chapter 15]{KT81} and the following lemma.
Introduce, for $y_\sigma > {\widehat{a}} >0$ and $x\in (0,{\widehat{a}}]$,
\begin{equation}\label{eq_def_s_hat}
\begin{array}{rlrlrl}
\widehat{s}(x) & := \displaystyle \exp\left\{ \int_x^{\widehat{a}} \frac{2\widehat{b}(y)}{\widehat{\sigma}^2(y)} \D y \right\},
 & \displaystyle 
\widehat{S}(0,{\widehat{a}}] & := \displaystyle \int_0^{\widehat{a}} \widehat s(y) \D y, & 
\displaystyle \widehat{M}(0,{\widehat{a}}] & := \displaystyle \int_0^{\widehat{a}} \frac{\D x}{\widehat{\sigma}^2(x)\widehat s(x)},\\
\displaystyle 
\widehat{N}(0) & := \displaystyle \int_0^{\widehat{a}} \frac{\widehat S[x,{\widehat{a}}]}{\widehat{\sigma}^2(x)\widehat s(x)} \D x, & 
 \displaystyle 
\widehat \Sigma(0) & := \displaystyle \int_0^{\widehat{a}} \frac{\widehat S(0,x]}{\widehat{\sigma}^2(x)\widehat s(x)}  \D x,
& \displaystyle S[x,{\widehat{a}}] & := \displaystyle 
\int_{x}^{\widehat{a}} \widehat s(y) \D y.
\end{array}
\end{equation}

\begin{lemma}\label{lemma_y_sig_2}
The following hold:
\begin{equation*}
\begin{array}{llll}
    \displaystyle  \widehat S(0,{\widehat{a}}]=\infty,
    &
    &\displaystyle  \widehat N(0)< \infty,
    &\displaystyle  \quad \text{ if }y_\sigma > 1,\\
    \displaystyle  \widehat S(0,{\widehat{a}}] < \infty,
    &
    &\displaystyle  \widehat \Sigma(0) = \infty,
    &\displaystyle  \quad \text{ if }y_\sigma = 1,\\
    \displaystyle  \widehat S(0,{\widehat{a}}] <\infty,
    &\displaystyle  \widehat{M}(0,{\widehat{a}}] = \infty, 
    &\displaystyle  \widehat \Sigma(0) < \infty,
    &\displaystyle  \quad \text{ if }y_\sigma < 1.
\end{array}
\end{equation*}
\end{lemma}

\begin{proof}[Proof of Lemma~\ref{lemma_y_sig_2}]
A straightforward Taylor expansion around the origin yields
\begin{equation}\label{eq:Expan_left_m}
    \left\{1- \left(1-\frac{y}{y_\sigma}\right)^{\beta} \right\}^{-2}
    = \frac{y^2_\sigma}{\beta^2 y^2}\left\{ 1 + \chi_1 \frac{y}{y_\sigma} + \chi_2 \frac{y^2}{y_\sigma^2} + \chi_3\frac{y^3}{y_\sigma^3} + \Oo(x^4) \right\},
\end{equation}
with $\chi_1, \chi_2$ and $\chi_3$ defined in~\eqref{eq:ABCbeta_1}.
We start with the behaviour of the function~$\widehat{s}$ and its integrated version.
Consider first the case $y_\sigma > 1$.
We split the range of possibilities into two possible intervals for ${\widehat{a}}$:
\begin{itemize}
\item[(i)] If ${\widehat{a}}< y_\sigma-1$,
then $y_\sigma-x\ge y_\sigma-{\widehat{a}} > 1$ and $b$ is negative on $[y_\sigma -{\widehat{a}}, y_\sigma -x]$. 
Then, for $x\in (0,{\widehat{a}}]$,
\begin{align*}
\int_x^{\widehat{a}} \frac{\widehat{b}(y)}{\widehat{\sigma}^2(y)} \D y &= \int_x^{\widehat{a}} \frac{-b(y_\sigma-y)}{\widetilde{\sigma}^2(y_\sigma-y)} \D y
= -\int_x^{\widehat{a}} \frac{b(y_\sigma -y)\D y}{(y_\sigma-y)^{2(1-\beta)} \gamma^2 \left(1-(1-\frac{y}{y_\sigma})^{\beta} \right)^2}\\
&= \frac{1}{h} \int_x^{\widehat{a}} \frac{(y_\sigma -y)^{2\beta -1} (y_\sigma - y - 1)}{\gamma^2 \left(1-(1-\frac{y}{y_\sigma})^{\beta} \right)^2} \D y\\
&\ge \frac{\left( y_\sigma -{\widehat{a}}\right)^{2\beta -1} \left( y_\sigma -{\widehat{a}} -1\right)}{\gamma^2 h} \int_{x}^{\widehat{a}}\left(1-\left(1-\frac{y}{y_\sigma}\right)^{\beta} \right)^{-2}\D y,
\end{align*}
as $\min_{y\in [x,{\widehat{a}}]} \left[(y_\sigma -y)^{2\beta-1} (y_\sigma - y-1)\right]
= \left( y_\sigma -{\widehat{a}}\right)^{2\beta -1} \left( y_\sigma -{\widehat{a}} -1\right)>0$. 
Indeed the map $y \mapsto y^{2\beta -1} (y-1)$ is increasing on $[y_\sigma -{\widehat{a}}, y_\sigma -x]$
because $y_\sigma -{\widehat{a}}>1$.
Noting that~\eqref{eq:Expan_left_m} is uniform on $[x,{\widehat{a}}]$, we obtain, as~$x$ approaches zero
\begin{align*}
\exp & \left(\frac{2\left( y_\sigma -{\widehat{a}}\right)^{2\beta -1} \left( y_\sigma -{\widehat{a}} -1\right)}{\gamma^2 h} \int_{x}^{\widehat{a}} \frac{\D y}{\left(1-(1-\frac{y}{y_\sigma})^{\beta} \right)^2} \right)\\
& = \exp\left(\frac{\tbrhK}{x}\right)
\E^{\tbrhK\tbrthKab}
x^{-\frac{\chi_1 \tbrhK}{y_\sigma}} (1+\Oo(x)),
\end{align*}
with $\tbrhK := \frac{2(y_\sigma -{\widehat{a}})^{2\beta -1} (y_\sigma -{\widehat{a}} -1) y^2_\sigma}{h \beta^2 \gamma^2} >0$ 
and 
$\tbrthKab := -\frac{1}{{\widehat{a}}} + \frac{\chi_1}{y_\sigma} \log ({\widehat{a}}) + \frac{\chi_2}{y^2_\sigma} {\widehat{a}}$, and therefore 
$\lim_{x \downarrow 0} \widehat{s} (x) = \infty$ and $\widehat{S}(0,{\widehat{a}}]=\int_0^{\widehat{a}} \widehat{s} (x) \D x = \infty$.

\item[(ii)] If $y_\sigma-1\le {\widehat{a}} <y_\sigma$, then 
for $x\in (0,{\widehat{a}}]$,
\begin{align}
\int_x^{\widehat{a}} \frac{\widehat{b}(y)}{\widehat{\sigma}^2(y)} \D y 
= \int_x^{y_\sigma-1} \frac{\widehat{b}(y)}{\widehat{\sigma}^2(y)} \D y 
+ \int_{y_\sigma-1}^{\widehat{a}} \frac{\widehat{b}(y)}{\widehat{\sigma}^2(y)} \D y.
\end{align}

Similarly to (i), one can prove that $\lim_{x\downarrow 0}  \int_x^{y_\sigma-1} \frac{\widehat{b}(y)}{\widehat{\sigma}^2(y)} \D y  = \infty$. 
Then, on $(y_\sigma-1,{\widehat{a}}]$, $\widehat{\sigma}$ is not null and is continuous, thus bounded; 
similarly, $\widehat{b}$ is negative and continuous, hence bounded on $(y_\sigma-1,{\widehat{a}}]$. 
Therefore $ \int_{y_\sigma-1}^{\widehat{a}} \frac{\widehat{b}(y)}{\widehat{\sigma}^2(y)} \D y <\infty$, for $x\in(0,{\widehat{a}}]$,
and
$$
\lim_{x\downarrow 0} \widehat{s}(x) = \exp\left\{ \int_{y_\sigma-1}^{\widehat{a}} \frac{\widehat{b}(y)}{\widehat{\sigma}^2(y)} \D y \right\} \exp\left\{  \int_x^{y_\sigma-1} \frac{\widehat{b}(y)}{\widehat{\sigma}^2(y)} \D y \right\} = \infty.
$$

\end{itemize}

Let us now describe the case $y_\sigma = 1$.
Using the Taylor expansion around zero
\begin{equation}\label{eq:Expan_left_1_m}
    \left[1- \left(1-y\right)^{\beta} \right]^{-2}
    = \frac{1}{\beta^2 y^2}\left[ 1 +\chi_1 y + \chi_2 y^2 + \chi_3 y^3 + \Oo(y^4) \right],
\end{equation}
with $\chi_1, \chi_2$ and $\chi_3$ as in~\eqref{eq:ABCbeta_1}, we have, as~$y$ approaches zero,
\begin{equation}
\label{eq:Expan_left_integrand_1_m}
\frac{2\widehat{b}(y)}{\widehat{\sigma}^2(y)}
= \frac{\tbrooK}{y} \left( 1 + \tbrchif y + \tbrchis y^2 + \tbrchit y^3 + \Oo(y^4) \right), 
\end{equation}
with $\tbrooK = -\frac{2}{h \beta^2 \gamma^2} <0$ and $\tbrhhKab := \log({\widehat{a}}) +\overline{\chi_1}{\widehat{a}} + \frac{\tbrchis}{2} {\widehat{a}}^2 +\frac{\tbrchit}{3}{\widehat{a}}^3$. 
Since the expansion is again uniform on $[x,{\widehat{a}}]$, we obtain
\begin{align}\label{eq:Expan_left_rho_1_m}
    \widehat{s} (x) 
    &= \exp\left\{2\int_x^{\widehat{a}} \frac{\widehat{b}(y)}{\widehat{\sigma}^2 (y)} \D y \right\}
    = \exp\left\{\tbrooK\left(-\log(x)-\tbrchif x +o(x)\right)\right\}\exp\left\{\tbrooK\tbrhhKab\right\}, \nonumber\\
    &= \exp\left\{\tbrooK\tbrhhKab\right\}x^{-\tbrooK}(1+\Oo(x)).
\end{align}
Since $\tbrooK <0$, we easily deduce that 
$\lim_{x \downarrow 0} \widehat{s} (x) = 0$, and $\widehat S(0,{\widehat{a}}]=\int_0^{\widehat{a}} \widehat{s} (x) \D x$ is finite. 

Consider now the case $y_\sigma < 1$.
Using~\eqref{eq:Expan_left_m}, we have, as~$y$ approaches zero,
\begin{equation}
\label{eq:Expan_left_integrand_m}
    \frac{2\widehat{b}(y)}{\widehat{\sigma}^2(y)}
    = -\frac{\tbrK}{y^2} \left( 1 + \tbrchif y+ \tbrchis y^2 + \tbrchit y^3 + \Oo(y^4) \right), 
\end{equation}
with $\tbrK = \frac{2y^{2\beta +1}_\sigma (1-y_\sigma)}{h \beta^2 \gamma^2} >0$, as $y_\sigma <1$. 
Since the expansion is uniform on $[x,{\widehat{a}}]$, we obtain
\begin{align}\label{eq:Expan_left_rho_m}
    \widehat{s} (x) 
    &= \exp\left\{2\int_x^{\widehat{a}} \frac{\widehat{b}(y)}{\widehat{\sigma}^2 (y)} \D y \right\}
    = \exp\left\{-\frac{\tbrK}{x}\left[1-\tbrchif x \log x\right]\right\} \exp\left\{-\tbrK\tbrhKab + \tbrchis \tbrK x + \Oo(x^2)\right\}, \nonumber\\
    &=\exp\left\{-\frac{\tbrK}{x} - \tbrK \tbrhKab\right\} x^{\tbrK\tbrchif} (1+\Oo(x)),
\end{align}
where $\tbrhKab=-\frac{1}{\widehat{a}}+\tbrchif \log(\widehat{a})+\tbrchif \widehat{a}$.
Since $\tbrK >0$, we easily deduce that 
$\lim_{x \downarrow 0} \widehat{s} (x) = 0$, and $\widehat{S}(0,{\widehat{a}}]=\int_0^{\widehat{a}} \widehat{s} (x) \D x$ is finite.

The middle statement (when $y_\sigma < 1$) in the lemma is straightforward.
When $x\in (0,{\widehat{a}}]$, 
$\int_x^{\widehat{a}} \frac{\widehat{b}(y)}{\widehat{\sigma}^2(y)} \D y = - \int_{y_\sigma -{\widehat{a}}}^{y_\sigma -x} \frac{b(y)}{\widetilde{\sigma}^2 (y)} \D y$.
Since $0 < y_\sigma -{\widehat{a}} < y_\sigma -x < 1$, $b$ is positive on $[y_\sigma -{\widehat{a}}, y_\sigma -x]$ and the above integral is therefore negative.
Hence, $\widehat{s}$ is bounded by $1$ on $(0,{\widehat{a}}]$, and
$$
\widehat M(0,{\widehat{a}}]
=\int_0^{\widehat{a}} \frac{1}{\widehat{s}(x)\widehat{\sigma}^2(x)} \D x 
\ge \int_0^{\widehat{a}} \frac{\D x}{\widetilde{\sigma}^2 (y_\sigma -x)} = \infty,
$$
using~\eqref{eq:Expan_left_m}, which concludes the proof.

The final integrals in the lemma are delicate. 
We start with the case $y_\sigma < 1$. 
Using~\eqref{eq:Expan_left_rho_m}, 
we obtain the asymptotic behaviour of $\widehat{S}(0,\cdot]$ around zero by integrating that of~$\widehat{s} (\cdot)$ around~$0$. Classical asymptotic expansions for integrals (note that the leading contribution arises at the right boundary of the integration domain) yield, after the change of variable $y \mapsto xz$,
\begin{equation}\label{eq:Expan_left_s_m}
\begin{aligned}\displaystyle
\widehat{S}(0,x] 
&= \int_0^x \widehat{s} (y) \D y
= \E^{-\tbrK \tbrhKab} x^{1+\tbrK\tbrchif} \int_0^1 \exp\left\{-\frac{\tbrK}{xz} \right\} z^{\tbrK\tbrchif} \D z, \\
&= \E^{-\tbrK \tbrhKab} x^{1+\tbrK\tbrchif} \exp\left\{-\frac{\tbrK}{x}\right\} \left(\frac{x}{\tbrK} + \Oo(x^2) \right), \qquad \text{as }x\downarrow 0.
\end{aligned}
\end{equation}

Therefore, combining~\eqref{eq:Expan_left_m},~\eqref{eq:Expan_left_rho_m} and~\eqref{eq:Expan_left_s_m}, we obtain
\begin{align}
    \frac{\widehat{S}(0,x]}{\widehat{s} (x) \widehat{\sigma}^2(x)}
    = x\left(\frac{1}{\tbrK} x + \mathcal{O}(x^2)\right) \frac{y^{2\beta}_\sigma}{\beta^2 \gamma^2} \frac{1}{x^2} (1+\mathcal{O}(x))
    =  \frac{y^{2\beta}_\sigma}{\beta^2 \gamma^2 \tbrK} (1+\Oo(x)), 
\end{align}
which is integrable on $(0,a]$ and concludes the proof.

In the case $y_\sigma = 1$,
exploiting~\eqref{eq:Expan_left_rho_1_m}, we obtain the asymptotic behaviour of $\widehat{S}(0,\cdot]$ around zero by integrating that of~$\widehat{s}(\cdot)$ around~$0$. 
Indeed, after the change of variable $y \mapsto xz$,
\begin{equation}\label{eq:Expan_left_s_1_m}
    \widehat{S}(0,x] 
    = \int_0^x \widehat{s} (y) \D y
    = \frac{\exp\left\{\tbrooK\tbrhhKab\right\}}{ x^{\tbrooK-1}}
    \int_0^1  z^{-\tbrooK} \D z
    = \exp\left\{\tbrooK\tbrhhKab\right\} x^{1-\tbrooK} (1+\Oo(x)),
\end{equation}
as~$x$ tends to zero.
Then, exploiting~\eqref{eq:Expan_left_1_m},~\eqref{eq:Expan_left_rho_1_m} and~\eqref{eq:Expan_left_s_1_m}, we obtain
\begin{align*}
\frac{\widehat{S}(0,x]}{\widehat{s} (x) \widehat{\sigma}^2 (x)}
= \frac{1}{\beta^2 \gamma^2}
\frac{1}{x} \left(1 +\Oo(x)\right),
\end{align*}
which is not integrable on $(0,{\widehat{a}}]$ and thus
$\widehat{\Sigma}(0)=\infty$.

Finally, we  move to the case $y_\sigma > 1$.
Since $\lim_{x\downarrow 0} \widehat{S}[x,{\widehat{a}}] =\infty$ and $\lim_{x\downarrow 0} \widehat{s}(x) \widehat{\sigma}^2(x) = \infty$, one needs to study the behaviour of the integrand around zero to conclude.
For $\delta >0$ such that $x < \delta < {\widehat{a}}$ and $x>0$,
$
\int_x^{\widehat{a}} \widehat{s} (y) \D y 
= \int_x^{\delta} \widehat{s} (y) \D y + \int_{\delta}^{\widehat{a}} \widehat{s} (y) \D y
$. Note that the second integral is convergent as the integral of a continuous function over a closed interval of~$\RR$.

Classical asymptotic expansions for integrals~\cite[Chapter 3.3]{Miller} and~\eqref{eq:Expan_left_rho_m}, yield, after mapping $y \mapsto zx$,
\begin{align}
    \int_x^{\delta} \widehat{s} (y) \D y
    = \E^{-\tbrK \tbrhKab} x^{1+\tbrK\tbrchif } \int_1^{\delta / x} \exp\left\{ -\frac{\tbrK}{xz} \right\} z^{\tbrK\tbrchif} \D z
    = \E^{-\tbrK \tbrhKab} x^{1+\tbrK\tbrchif} \E^{-\frac{\tbrK}{x}} \left(-\frac{x}{\tbrK} + \Oo(x^2)\right),
\end{align}
and the asymptotic behaviour of the integrand around the origin is given by
\begin{align}
    \frac{\widehat{S}[x,{\widehat{a}}]}{\widehat{s} (x) \widehat{\sigma}^2(x)}
    = -\frac{y^{2\beta}_\sigma}{\beta^2 \gamma^2} \left(\frac{x^2}{\tbrK} + \Oo(x^3)\right) \frac{1+\Oo(x)}{x^2}
    = -\frac{y^{2\beta}_\sigma}{\beta^2 \gamma^2 \tbrK} (1+ \Oo(x)),
\end{align}
which is integrable at the origin, and concludes the proof since $ \widehat{N}(0)$ is therefore finite.
\end{proof}

\section{Ergodicity proofs}

\subsection{Proof of Theorem~\ref{thm:ErgodicB}}\label{Proofthm:ErgodicB}

This study is based on~\cite[Theorem 1.1, Chapter 5.1]{Pinsky}.
Since~$\widetilde{\sigma}$ is null at $y_\sigma$ (and possibly at zero), we consider separately the two domains $\Dom_1:= (0, y_\sigma)$ and $\Dom_2:= (y_\sigma, \infty)$
so that Assumption~A~iii) in the aforementioned theorem is satisfied.  
We start with $y_0 \in \Dom_2= (y_\sigma, \infty)$. We have to check the finiteness of 
$$
\bold{A}:= \int_{y_\sigma}^{y_0} \exp\left\{ -\int_{y_0}^y \frac{2b(s)}{\widetilde{\sigma}^2(s)}\D s\right\} \D y
\quad \text{and} \quad
\bold{B}:= \int_{y_0}^{\infty} \exp\left\{ -\int_{y_0}^y \frac{2b(s)}{\widetilde{\sigma}^2(s)}\D s\right\} \D y.
$$
Starting with $\bold{A}$, the changes of variables $y \to x+y_\sigma$ and $s \to v+y_\sigma$ yield
\begin{equation}
\bold{A} = \int_{0}^{y_0-y_\sigma} \exp\left\{ \int_{x}^{y_0-y_\sigma} \frac{2b(v+y_\sigma)}{\widetilde{\sigma}^2(v+y_\sigma)}\D v\right\} \D x
= \int_{0}^{y_0-y_\sigma} \exp\left\{ \int_{x}^{y_0-y_\sigma} \frac{2\overline{b}(v)}{\overline{\sigma}^2(v)}\D v\right\} \D x,
\end{equation}
with $\overline{b}(v):= b(v+y_\sigma)$ and $\overline{\sigma}(v):= \widetilde{\sigma}(v+y_\sigma)$, $v \geq 0$.
Notice that~$\bold{A}=\int_0^a \overline{s}(x) \D x$, with $ \overline{s}$ as in~\eqref{eq_def_s_bar} and $a:= y_0-y_\sigma \geq 0$. Thus, exploiting the proof of Lemma~\ref{lem_y_sigma_KT}, 
 $\bold{A}$ is finite if and only if $y_\sigma \geq 1$.
Regarding~$\bold{B}$, we need to study the integrand at infinity.
Since,
as $s \uparrow\infty$,
\begin{equation}
\frac{2b(s)}{\widetilde{\sigma}^2(s)} = \frac{2(1-s)}{hs\left(-\frac{\alpha}{\beta}+\gamma s^{-\beta}\right)^2} \sim -\frac{2}{h}\left(\frac{\beta}{\alpha}\right)^2,
\end{equation}
then
$
\int_{y_0}^y\frac{2b(s)}{\widetilde{\sigma}^2(s)} \D s
\sim -\frac{2\beta^2}{h \alpha^2} y,
$
as~$y\uparrow\infty$, so~$\bold{B}$ is infinite, concluding the $y_0 \in \Dom_2$ discussion.
\\
Consider now the domain $\Dom_1=(0,y_\sigma)$. We have to check the finiteness of 
$$
\bold{C}:= \int_{0}^{y_0} \exp\left\{ -\int_{y_0}^y \frac{2b(s)}{\widetilde{\sigma}^2(s)}\D s\right\} \D y,
\quad \text{and} \quad
\bold{D}:= \int_{y_0}^{y_\sigma} \exp\left\{ -\int_{y_0}^y \frac{2b(s)}{\widetilde{\sigma}^2(s)}\D s\right\} \D y.
$$
For~$\bold{C}$, the only possible issue for integrability is at zero, so we expand the integrand in a neighborhood of the origin:
\begin{equation*}
\frac{2b(s)}{\widetilde{\sigma}^2(s)} 
= \frac{2(1-s)}{h\gamma^2 s^{1-2\beta}}\left( 1- \frac{\alpha}{\beta \gamma} s^{\beta} \right)^{-2}
= \frac{2 s^{2\beta-1}}{h\gamma^2}
\left[1 + \frac{2\alpha}{\beta \gamma} s^{\beta} +3 \left(\frac{\alpha}{\beta \gamma}\right)^2 s^{2\beta} + \Oo\left(s^{3\beta}\right) \right].
\end{equation*}
Since the expansion is uniform on $(0,y_0)$,
\begin{equation*}
\int_{y_0}^{x} \frac{2b(s)}{\widetilde{\sigma}^2(s)} \D s
= \frac{2}{\beta h \gamma^2}\left(\frac{x^{2\beta}}{2} + \frac{2\alpha}{3\beta \gamma}x^{3 \beta} + \frac{3\alpha^2}{4\beta^2\gamma^2}x^{4\beta}
+ o\left(x^{4 \beta}\right)+ \tbrtKyb\right),
\end{equation*}
with $\tbrtKyb:= -\frac{y_0^{2\beta}}{2} - \frac{2\alpha}{3\beta \gamma}y_0^{3 \beta} - \frac{3\alpha^2}{4\beta^2\gamma^2}y_0^{4 \beta}$.
We then obtain
\begin{equation*}
\exp\left\{-\int_{y_0}^{x} \frac{2b(s)}{\widetilde{\sigma}^2(s)} \D s\right\}
= \exp\left\{-\frac{x^{2\beta} + o\left(x^{2 \beta}\right)}{\beta h \gamma^2} -\frac{2 \tbrtKyb}{\beta h \gamma^2}\right\}
= \exp\left\{-\frac{2 \tbrtKyb}{\beta h \gamma^2}\right\}
\left(1+o\left(x^{2\beta}\right)\right),
\end{equation*}
and so $\bold{C}$ is always finite.
Finally, regarding~$\bold{D}$, the following sequence of change of variables, 
as $x\downarrow y_\sigma-y$ and $s\downarrow y_\sigma-v$, yields
\begin{equation*}
\begin{aligned}
\bold{D}
&= \int_{y_0}^{y_\sigma} \exp\left\{ -\int_{y_0}^x \frac{2b(s)}{\widetilde{\sigma}^2(s)}\D s\right\} \D x
=  \int_{0}^{y_\sigma-y_0} \exp\left\{ -\int_{y}^{y_\sigma-y_0} \frac{2b(y_\sigma-v)}{\widetilde{\sigma}^2(y_\sigma-v)}\D v\right\} \D y \\
& =  \int_{0}^{y_\sigma-y_0} \exp\left\{ \int_{y}^{y_\sigma-y_0} \frac{2\widehat{b}(v)}{\widehat{\sigma}^2(y_\sigma-v)}\D v\right\} \D y, 
\end{aligned}
\end{equation*}
with $\widehat{b}(v)= -b(y_\sigma-v)$ and $\widehat{\sigma}(v)= -\widetilde{\sigma}(y_\sigma-v)$, for $v \in (0,y_\sigma)$.
Now, notice that $\bold{D}=\int_0^a \widehat{s}(x) \D x$, with $\widehat{s}$ defined in~\eqref{eq_def_s_bar} and $a=y_\sigma-y_0$. Thus, exploiting the proof of Lemma~\ref{lemma_y_sig_2}, 
the integral~$\bold{D}$ is finite if and only if $y_\sigma \leq 1$, and the theorem follows.

\subsection{Proof of Proposition~\ref{prop:Stationary}}
\label{prop:Stationary_proof}
Pursuing the analysis in~\cite[page 242]{KT81}, 
we can prove that
\begin{displaymath}
\begin{aligned}
S_{\Pi}(y_\sigma,x]& := \int_{y_\sigma}^{x} s(\xi) \D \xi= \infty, 
& S_{\Pi}[x, \infty)& := \int_{x}^{\infty} s(\xi) \D \xi=\infty,\\
M_{\Pi}(y_\sigma,x]& := \int_{y_\sigma}^{x} \frac{\D \xi}{\widetilde{\sigma}^2(\xi)s(\xi)} < \infty,
& M_{\Pi}[x, \infty)& := \int_{x}^{\infty} \frac{\D \xi}{\widetilde{\sigma}^2(\xi)s(\xi)} < \infty,
\end{aligned}
\end{displaymath}
where $s(\xi):=\exp\left\{ -\int_a^{\xi} \frac{2b(\eta)}{\widetilde{\sigma}^2(\eta)} \D \eta \right\}$, with $a > y_\sigma$ fixed.
Thus, the stationary probability measure is  given by~\eqref{eq:StationaryY}.
Let us start by showing that $S_{\Pi}(y_\sigma,x]$ and $S_{\Pi}[x, \infty)$ are infinite. Exploiting the computations in the proof of Theorem~\ref{thm:ErgodicB}, for the case $y_\sigma < 1$ and $\Dom=(y_\sigma, \infty)$, we obtain the unboundedness of
\begin{displaymath}
\begin{aligned}
	& S_{\Pi}(y_\sigma,x]
	= \int_{y_\sigma}^{x} s(\xi) \D \xi
	= \int_0^{x-y_\sigma} \exp\left\{ -\int_y^{x-y_\sigma} \frac{2b(\eta+y_\sigma)}{\widetilde{\sigma}^2(\eta+y_\sigma)} \D \eta \right\} \D y
	= \int_0^{x-y_\sigma}\overline{s}(y) \D y, 
\end{aligned}
\end{displaymath}
and
\begin{displaymath}
\begin{aligned}
	& S_{\Pi}[x, \infty)
	= \int_{x}^{\infty} s(\xi) \D \xi
	= \int_{x}^{\infty} \exp\left\{ -\frac{2}{h}\int_x^{\xi} \frac{(1-\eta)\D \eta}{\eta(-\frac{\alpha}{\beta}+\gamma \eta^{-\beta})} \right\} \D y
	\sim \int_{x}^{\infty} \exp\left\{ \frac{2}{h}\int_x^{\xi} \frac{\beta}{\alpha}\eta \D \eta \right\} \D y.
\end{aligned}
\end{displaymath}
We are only left to prove that $M_{\Pi}(y_\sigma,x]$ and $M_{\Pi}[x, \infty)$ are finite. 
Exploiting the changes of variables $\xi= v+y_\sigma$ and $\eta=z+y_\sigma$, $M_{\Pi}(y_\sigma,x]$  can be rewritten as
\begin{displaymath}
\begin{aligned}
	M_{\Pi}(y_\sigma,x] &
	= \int_{y_\sigma}^{x} \frac{\D \xi}{\widetilde{\sigma}^2(\xi)s(\xi)} 
	= \int_{y_\sigma}^{x} \frac{1}{\widetilde{\sigma}^2(\xi)}
	\exp\left\{ \int_x^{\xi} \frac{2b(\eta)}{\widetilde{\sigma}^2(\eta)} \D \eta \right\} \D \xi\\
	& = \int_{0}^{x-y_\sigma} \frac{1}{\widetilde{\sigma}^2(v+y_\sigma)}
	\exp\left\{ \int_{x-y_\sigma}^{v} \frac{2b(z+y_\sigma)}{\widetilde{\sigma}^2(z+y_\sigma)} \D z \right\} \D v\\
	& = \int_{0}^{x-y_\sigma} \frac{\D v}{\widetilde{\sigma}^2(v+y_\sigma) \overline{s}(v)}
	= \int_{0}^{x-y_\sigma} \frac{\D v}{\overline{\sigma}^2(v) \overline{s}(v)}
	= \overline{M}(0,x-y_\sigma].
\end{aligned}
\end{displaymath}
Then, in a neighborhood of zero we have
\begin{displaymath}
\begin{aligned}
	\frac{1}{\overline{\sigma}^2(v) \overline{s}(v)}  
	= \frac{y_\sigma^{2 \beta}}{\beta^2 \gamma^2}\E^{-\tbrK \tbrKxyb}
	\exp\left\{-\frac{\tbrK}{v}\right\} v^{\tbrK \tbrchif -2}(1+\Oo(v)), 
\end{aligned}
\end{displaymath}
which is integrable around zero since $\tbrK >0$ and thus $M_{\Pi}(y_\sigma, x]$ is finite.

To conclude we have to study the finiteness of $M_{\Pi}[x, \infty)$, which means that we have to check the integrability of $\frac{1}{\widetilde{\sigma}^2(\xi)s(\xi)}$ at infinity.
Since 
\begin{displaymath}
\begin{aligned}
	\frac{1}{\widetilde{\sigma}^2(\xi)s(\xi)} 
	= \frac{1}{\widetilde{\sigma}^2(\xi)}
	\exp\left\{ \int_x^{\xi} \frac{2b(\eta)}{\widetilde{\sigma}^2(\eta)} \D \eta \right\}
	\sim \frac{\xi^{2\beta-2}}{\gamma} \exp\left\{ \int_x^\xi -\frac{2}{h}(1+\Oo(\eta)) \D \eta \right\}
	\sim \frac{\xi^{2\beta-2}}{\gamma} \E^{ -\frac{2}{h} \xi}
\end{aligned}
\end{displaymath}
is integrable at infinity, the result follows.

\subsection{Proof of Proposition~\ref{prop:PerturPriceExpansion}}
\label{prop:PerturPriceExpansion_Proof}
The aim is to approximate the price of an option with smooth payoff~$h$ 
as $P^\eps\approx Q^\eps:= P_0+\sqrt{\eps}P_1$, that is $P^\eps= Q^\eps + \mathcal{O}(\eps)$. 
We show that both~$P_0$ and~$P_1$ in fact do not depend on~$y$
and provide a precise estimate for the error term.
A Taylor expansion of~$P^\eps$ around $\eps=0$ gives
$$
P^\eps 
= P_0+\sqrt{\eps}P_1 + \eps Q_2 + \eps^{3/2}Q_3+ \Oo\left(\eps^{3/2}\right)
= Q^\eps + \eps Q_2 + \eps^{3/2}Q_3+ \Oo\left(\eps^{3/2}\right).
$$
The pricing PDE~\eqref{eq:pricing_eps} then reads
\begin{align*}
0 & = 
\left(\frac{1}{\eps}\Ll_Y+\frac{1}{\sqrt{\eps}}\Ll_1+\Ll_{\BS}^{\sigma(y)}\right) P^\eps\\
 & = \left(\frac{1}{\eps}\Ll_Y+\frac{1}{\sqrt{\eps}}\Ll_1+\Ll_{\BS}^{\sigma(y)}\right)
 \left(P_0+\sqrt{\eps}P_1 + \eps Q_2 + \eps^{3/2}Q_3+ \Oo\left(\eps^{3/2}\right)\right)\\
 & = \frac{\Ll_Y P_0}{\eps}
 + \frac{\Ll_Y P_1 + \Ll_1 P_0}{\sqrt{\eps}}
 + \Big[\Ll_Y Q_2 +\Ll_1P_1+\Ll_{\BS}^{\sigma(y)} P_0\Big]
+ \sqrt{\eps}\left[\Ll_Y Q_3 + \Ll_1 Q_2 +\Ll_{\BS}^{\sigma(y)} P_1\right]
+ \Oo(\eps).
\end{align*}
Since this should be null for all (small)~$\eps$, each term should be equal to zero.
More specifically,
\begin{enumerate}
\item[a.]  $\Ll_Y P_0=0$.
Since~$\Ll_Y $ has no $x$-derivative, $P_0(t,x,y)= P_0(t,x)$ with $P_0(T,x)=h(x)$;

\item[b.]  $\Ll_Y P_1+\Ll_1P_0=0 =\Ll_Y P_1$
using~a..
Similarly $P_1(t,x,y)= P_1(t,x)$ with $P_1(T,x)=0$;

\item[c.]  $0 = \Ll_Y Q_2 +\Ll_1P_1+\Ll_{\BS}^{\sigma(y)} P_0
 = \Ll_Y Q_2 +\Ll_{\BS}^{\sigma(y)} P_0$.
This is a Poisson equation associated to~$\Ll_Y $
and requires a suitable solvability condition:
Similarly to~\cite{FLS16}, the Fredholm alternative\footnote{As far as we know, there is no general Fredholm alternative for hypoelliptic operators. Numerical tests seem to clearly indicate the presence of a spectral gap in our case, which would be enough, but we leave this very lengthy and detailed analysis to further research.} imposes the condition
\begin{align*}
0 = \langle \Ll_{\BS}^{\sigma} P_0 , \Pi\rangle
& = \int_{y_\sigma}^\infty \left(\Ll_{\BS}^{\sigma} P_0\right)
\Pi(\D y) 
 = \int_{y_\sigma}^\infty \left[\partial_{t} +  \frac{\sigma^2(y)}{2}
\Ddx\right] P_0(t,x) \Pi(\D y)\\
 & =\left[\partial_{t} + \half\int_{y_\sigma}^\infty  \sigma^2(y)\Pi(\D y) \right]\Ddx P_0(t,x),
\end{align*}
where $\Pi$ is the unique stationary distribution of~$Y$ on $\Dom=(y_\sigma, \infty)$
from Proposition~\ref{prop:Stationary},
and with the operator
$\Ddx := \dxx - \dx$.
\end{enumerate}
This last computation in particular reveals that
\begin{equation}\label{eq:BSLim}
\Lg \Ll_{\BS}^{\sigma}, \Pi \Rg = \LBsLim,
\end{equation}
so that~$P_0$ in fact satisfies
$\Ll_{\BS}^{\sigmaLim}P_0(t,x)=0$,
with boundary condition $P_0(T,x)=h(x)$,
so that~$P_0$ corresponds to the Black-Scholes option price
with payoff~$h$
and variance 
$\sigmaLim^2 : = \langle \sigma^2 , \Pi \rangle =  \int_{y_\sigma}^\infty \sigma^2(y)\Pi(\D y)$,
as given in the proposition.
\begin{remark}\label{Rem_fin_sigma_Pi}
The variance~$\sigmaLim^2$ is clearly finite:
using the asymptotic computations
in the previous section,
as $y \uparrow\infty$, the integrand behaves as $\exp\{ - \frac{2}{h}(\frac{\alpha}{\beta})^2 y \}y^{-2}$ which is integrable at infinity~\eqref{Eq_bb_inf_int_s}. 
When $y\downarrow y_\sigma$ it  behaves as $\exp\{-\frac{K}{(y-y_\sigma)}-KK^\alpha_\beta\}(y-y_\sigma)^{-\tbrchif-2}(1+\Oo(y-y_\sigma))$, also integrable since $K>0$~\eqref{eq:Expan_right_rho_m}.
\end{remark}
Observe now from~c. above that 
\begin{equation}\label{Eq_Q2_eq_1}
Q_2= -\Ll_Y^{-1}
\left(\Ll_1 P_1  + \Ll_{\BS}^{\sigma(y)} P_0\right)
= -\Ll_Y^{-1}
\left(\Ll_{\BS}^{\sigma(y)} P_0\right)
= -\Ll_Y^{-1}
\left(\Ll_{\BS}^{\sigma(y)} - \LBsLim\right)  P_0.
\end{equation}
Note that we do not formally need to invert~$\Ll_Y$, but it makes the notations below clearer.

\begin{itemize}
\item[d.] Regarding the $\sqrt{\eps}$ term, $\Ll_Y Q_3 +\Ll_1 Q_2 +\Ll_{\BS}^{\sigma(y)} P_1=0$.
This is again a Poisson equation with solvability condition (using~\eqref{Eq_Q2_eq_1})
\begin{equation}\label{eq:solv_cond_2}
\LBsLim P_1 = 
\Lg \Ll_{\BS}^{\sigma}, \Pi \Rg P_1
= - \langle \Ll_1 Q_2, \Pi \rangle 
= \Lg \Ll_1\Ll_Y^{-1}
\left(\Ll_{\BS}^{\sigma} - \LBsLim\right), \Pi \Rg P_0.
\end{equation}
\end{itemize}
Combining this with the terminal condition on~$P_1$ obtained in~b., we obtain
\begin{equation}
\label{eq:BS_mod_sig_w_source}
\left\{
\begin{array}{ll}\displaystyle
\LBsLim P_1(t,x)
=\left\langle \Ll_1\Ll_Y^{-1}\left(\Ll_{\BS}^{\sigma} - \LBsLim\right), \Pi\right\rangle P_0, \\
P_1(T,x)=0,
\end{array}
\right.
\end{equation}
so that~$P_1$ is the solution of a Black-Scholes system with variance $\sigmaLim^2$ and source
\begin{equation}
\Lg \Ll_1\Ll_Y^{-1}
\left(\Ll_{\BS}^{\sigma} - \LBsLim\right), \Pi \Rg P_0 
 = \half
 \Lg \Ll_1\Ll_Y^{-1}
 \left(\sigma^2 - \sigmaLim^2\right), \Pi \Rg
 \Ddx P_0.
\end{equation}
Setting~$\psi$ to be the solution 
to $\Ll_Y \psi(y)= \sigma^2(y) - \sigmaLim^2$
in~\eqref{Eq_eq_def_psi},
we obtain
$$
\Lg \Ll_1\Ll_Y^{-1}\left(\Ll_{\BS}^{\sigma} - \LBsLim\right), \Pi \Rg P_0
 =\half\Lg \Ll_1\psi(\cdot), \Pi \Rg \Ddx P_0
 =\half
\Lg \varpi, \Pi \Rg \dx\Ddx P_0,
$$
by definition of~$\Ll_1$ in~\eqref{eq:Operators}  and of~$\varpi$ in~\eqref{Eq_eq_def_varpi}.
The last term on the right-hand side is well defined provided that~\eqref{Eq_eq_def_psi} admits a unique solution such that 
$\Lg \varpi, \Pi \Rg$ is finite.
The existence of such a unique (up to some positive constant) solution is ensured by the validity of the corresponding solvability condition, consequence of $\sigmaLim^2$ being finite 
(as proved in Remark~\ref{Rem_fin_sigma_Pi}). 
A similar argument shows that 
$\Lg \varpi, \Pi \Rg$ is also finite 
once we prove polynomial growth at infinity and the boundedness of~$\psi$ around~$y_\sigma$. Indeed, in that case, for $y \uparrow\infty$, the integrand behaves like $\exp \{ - \frac{2}{h}(\frac{\alpha}{\beta})^2 y \}y^{-1}\psi'(y) \sim \exp\{ - \frac{2}{h}(\frac{\alpha}{\beta})^2 y\}(1+y^{n-1})$, which is integrable. 
As $y\downarrow y_\sigma$, we have $\exp\{-\frac{K}{(y-y_\sigma)}-KK^\alpha_\beta\}(y-y_\sigma)^{-K\tbrchif-1}\psi'(y)(1+\Oo(y-y_\sigma))$, which is integrable since $K>0$.
We thus conclude that
\begin{equation}\label{Eq_for_P1}
\sqrt{\eps}P_1(t,x) = -(T-t)\Omega^{\eps}\dx\Ddx P_0(t,x),
\end{equation}
with $\Omega^{\eps}:=\frac{\sqrt{\eps}}{2}\Lg \varpi, \Pi \Rg$. 
This implies
$P_1(T,x)=0$
and, since $\langle \Ll_{\BS}^{\sigma}, \Pi\rangle P_0 = \LBsLim P_0 =0$,
\begin{equation*}
\begin{aligned}
\LBsLim P_1 &
= \frac{1}{\sqrt{\eps}}
\LBsLim \Big( -(T-t)\Omega^{\eps} \dx\Ddx P_0 \Big)\\
& =\frac{\Omega^{\eps} \dx\Ddx P_0 -(T-t) \Omega^{\eps} 
\dx\Ddx \LBsLim P_0}{\sqrt{\eps}}
 =\frac{\Omega^{\eps}}{\sqrt{\eps}} \dx\Ddx P_0 
= \left\langle \Ll_1\Ll_Y^{-1}
\left(\Ll_{\BS}^{\sigma} - \LBsLim\right), \Pi \right\rangle P_0,
\end{aligned}
\end{equation*}
which corresponds precisely to~\eqref{eq:solv_cond_2}.

We now move on to the proof of the error term, 
assuming a smooth payoff~$h$.
With $\Ll_\eps := \frac{1}{\eps}\Ll_Y +  \frac{1}{\sqrt \eps}\Ll_1 + \Ll^{\sigma(y)}_{\BS}$
and 
\begin{equation}
Z^{\eps}:=
\eps Q_2 + \eps \sqrt{\eps} Q_3 - \left(P^\eps-Q^\eps\right)
= \eps Q_2 + \eps \sqrt{\eps} Q_3
- \left[P^\eps - 
\left(P_0+P_1 \sqrt{\eps}\right)\right],
\end{equation}
the pricing PDE~\eqref{eq:pricing_eps} now yields
\begin{align*}
\Ll_\eps Z^\eps &
= \Ll_\eps 
\left(\eps Q_2 + \eps \sqrt{\eps} Q_3
- \left[P^\eps
 - \left(P_0+P_1 \sqrt{\eps}\right)\right]\right)
= \Ll_\eps 
\left(P_0+P_1 \sqrt{\eps}+\eps Q_2 + \eps \sqrt{\eps} Q_3 -P^\eps\right)\\
& = \frac{\Ll_Y P_0}{\eps} + \frac{\Ll_YP_1+\Ll_1P_0}{\sqrt{\eps}} + \left(\Ll_Y Q_2+\Ll_1P_1+ \Ll_{\BS}^{\sigma(y)}P_0\right)
 + \sqrt{\eps}
 \left(\Ll_Y Q_3+\Ll_1Q_2+ \Ll_{\BS}^{\sigma(y)}P_1\right)\\
& \qquad +\eps
\left(\Ll_1Q_3+ \Ll_{\BS}^{\sigma(y)}Q_2\right) +\eps\sqrt{\eps} \Ll_{\BS}^{\sigma(y)}Q_3 -\Ll_\eps P^\eps\\
& = \eps
\left(\Ll_1Q_3+ \Ll_{\BS}^{\sigma(y)}Q_2 +\sqrt{\eps} \Ll_{\BS}^{\sigma(y)}Q_3\right).
\end{align*}

Setting
\begin{equation}\label{Eq_def_Feps_Geps}
\left\{
\begin{array}{rl}
\displaystyle F_\eps(t,x,y) & := \displaystyle \Ll_1Q_3+ \Ll_{\BS}^{\sigma(y)}Q_2 +\sqrt{\eps} \Ll_{\BS}^{\sigma(y)}Q_3\\
\displaystyle G_\eps(x,y) & := \displaystyle Q_2(T,x,y) +\sqrt{\eps}Q_3(T,x,y),
\end{array}\right.
\end{equation}
we write a parabolic PDE associated to $Z^{\eps}$:
\begin{equation}\label{Eq_PDE_Z^vareps}
\Ll_\eps Z^\eps = \eps F_\eps,
\qquad\text{with boundary condition}\qquad
Z^\eps(T,x,y)= \eps G_\eps(x,y).
\end{equation}
The first is a consequence of the identities
above, while the second follows from
\begin{align*}
	Z^\eps(T,x,y)&
	=  \eps Q_2(T,x,y) + \eps \sqrt{\eps} Q_3(T,x,y) -[P^\eps(T,x,y)-(P_0(T,x,y)+P_1(T,x,y) \sqrt{\eps})]\\ \nonumber
	& = \eps Q_2(T,x,y) + \eps \sqrt{\eps} Q_3(T,x,y) -[h(x)-(h(x)+0)]\\ \nonumber
	& =  \eps Q_2(T,x,y) + \eps \sqrt{\eps} Q_3(T,x,y)
	= \eps G_\eps(x,y)
\end{align*}
We now investigate the form of~$Q_2, Q_3$.
From the third identity in~\eqref{Eq_Q2_eq_1}, 
$Q_2 = -\half\psi(y) \Ddx P_0$,
where~$\psi$ is the solution to~\eqref{Eq_eq_def_psi},
which implies (recall that~$P_0$ does not depend on~$y$)
\begin{align*}
	\Ll_Y Q_2
	& = -\half\Ll_Y \left[ \psi(y) \Ddx P_0\right]
	 = -\frac{\sigma^2(y) - \sigmaLim^2}{2} \Ddx P_0
	= -\left(\Ll_{\BS}^{\sigma(y)} - \LBsLim\right)P_0
	 = -\Ll_{\BS}^{\sigma(y)}P_0	
\end{align*} 
The core idea here is to rewrite~$F_\eps$ and~$G_\eps$ to obtain the order of convergence of the first-order price approximation $Q^\eps$.
In the following computations, $P_0$ has smooth derivatives since the payoff~$h$ is smooth by assumption.
The identity
\begin{equation}\label{Eq_def_LBS}
\Ll_{\BS}^{\sigma(y)}
 = \partial_{t} + \frac{\sigma^2(y)}{2}\Ddx = \Ll_{\BS}^{\sigmaLim} + \frac{\sigma^2(y)-\sigmaLim^2}{2}\Ddx
\end{equation}
holds, yielding an explicit expression for the second term on the right-hand side of~\eqref{Eq_def_Feps_Geps}:
\begin{align}\label{eq_term_2_in_Feps}
	\Ll_{\BS}^{\sigma(y)} Q_2 
	 = \left(\LBsLim  + \frac{\sigma^2(y)-\sigmaLim^2}{2}\Ddx\right)
	\left( -\frac{\psi(y)}{2}\Ddx P_0\right) 
	& = -\frac{\psi(y)}{2}
	\left(\LBsLim  \Ddx P_0 + \frac{\sigma^2(y)-\sigmaLim^2}{2}\Ddx^2 P_0 \right) \nonumber\\
	& = -\frac{\sigma^2(y)-\sigmaLim^2}{4}\psi(y)\Ddx^2 P_0, 
\end{align}
since these differential operators commute.
Now, $Q_3$ is solution to the Poisson equation
$\Ll_Y Q_3 = - \left(\Ll_1 Q_2 + \Ll_{\BS}^{\sigma(y)}P_1\right)$,
and the validity of the centering condition for the Poisson equation is guaranteed by the choice of~$P_1$.
Equivalently,
\begin{align}\label{Eq_def_Q3}
	Q_3
	& = - \Ll_Y^{-1}\left(\Ll_1 Q_2 + \Ll_{\BS}^{\sigma(y)} P_1\right) \nonumber\\
	& = - \Ll_Y^{-1}\left(\Ll_1 Q_2 + \Ll_{\BS}^{\sigma(y)} P_1- \Lg \Ll_1 Q_2 +\Ll_{\BS}^{\sigma} P_1, \Pi\Rg\right) \nonumber\\
	& = - \Ll_Y^{-1}
	\left(\Ll_1 Q_2 - \Lg \Ll_1 Q_2, \Pi\Rg + \left(\Ll_{\BS}^{\sigma(y)} - \LBsLim\right) P_1\right).
\end{align}
We make the terms on the right more explicit
\begin{align}\label{Eq_def_L1Q2}
	\Ll_1 Q_2 = -\left(y\sigma^2(y) \partial_{xy}\right)
	\left(\frac{\psi(y)}{2}\Ddx P_0\right)
	& =  -\frac{y\sigma^2(y)}{2}\dx \Big(\psi'(y)\Ddx P_0 + \psi(y)\Ddx \partial_{y} P_0\Big)\nonumber\\
	& =  -\frac{\varpi(y)}{2}\dx \Ddx P_0 .
\end{align}
Now, let~$\vartheta$ be the solution to the Poisson equation
\begin{equation}\label{Eq_eq_vartheta_def}
\LY \vartheta = \varpi(y)
 - \Lg  \varpi, \Pi \Rg,
\end{equation}
and plug~\eqref{Eq_def_LBS} and~\eqref{Eq_def_L1Q2} into~\eqref{Eq_def_Q3} to obtain
\begin{align}\label{eq_Q3 _expl}
Q_3 
& = -\LY^{-1}\left( -\frac{\varpi(y)}{2}\dx \Ddx P_0
+ \frac{\Lg \varpi, \Pi \Rg}{2}
\dx \Ddx P_0
+ \frac{\sigma^2(y)-\sigmaLim^2}{2}\Ddx P_1\right)\\\nonumber
& = \half \LY^{-1}\left[ \Big(\varpi(y)
 - \Lg \varpi, \Pi \Rg\Big) \dx \Ddx P_0 -  \left(\sigma^2(y)-\sigmaLim^2\right)\Ddx P_1\right]\\ \nonumber
& = \half \LY^{-1}\Big( \LY \vartheta\dx \Ddx P_0 -  \LY \psi \Ddx P_1\Big) 
= \half\Big( \vartheta\dx \Ddx P_0 -  \psi \Ddx P_1\Big).
\end{align}
Exploiting the definition of $\Ll_1$, we obtain the first term in the expansion for $F_\eps$:
\begin{align}\label{Eq_term2_exp_Fvareps}
	\Ll_1 Q_3 
	& = \frac{y\sigma^{2}(y)}{2} \partial_{xy} \Big( \vartheta\dx \Ddx P_0 -  \psi \Ddx P_1\Big)\nonumber\\
	  & = \frac{y\sigma^{2}(y)}{2}
	 \Big( \vartheta'(y)\dx^2 \Ddx P_0 -  \psi'(y) \dx \Ddx P_1\Big).
\end{align}
Finally, exploiting~\eqref{Eq_def_LBS}-\eqref{eq_Q3 _expl}, together with~\eqref{Eq_for_P1}, we write
\begin{align}\label{eq_term_3_in_Feps}
	\Ll_{\BS}^{\sigma(y)} Q_3 
	& = 
	\half\left(\LBsLim +\frac{\sigma^2(y)-\sigmaLim^2}{2}\Ddx\right)
	\left( \vartheta\dx \Ddx P_0 -  \psi \Ddx P_1\right)\nonumber \\
	& = -\frac{\psi(y)}{2} \widetilde{\Omega} \dx \Ddx^2 P_0
	+\frac{\sigma^2 (y)-\sigmaLim^2}{4} \vartheta(y)\dx\Ddx^2 P_0
	- \frac{\sigma^2 (y)-\sigmaLim^2}{4}\psi(y)\Ddx^2 P_1,
\end{align}
where $\widetilde{\Omega}
:= \frac{1}{\sqrt{\eps}}\Omega^{\eps}
= \half \langle \varpi, \Pi \rangle $.

Placing~\eqref{Eq_term2_exp_Fvareps}-\eqref{eq_term_2_in_Feps}-\eqref{eq_term_3_in_Feps} in~\eqref{Eq_def_Feps_Geps}, we then obtain
\begin{align*}
	F_\eps(t,x,y)
	& = \Ll_1Q_3 + \Ll_{\BS}^{\sigma(y)}Q_2 +\sqrt{\eps} \Ll_{\BS}^{\sigma(y)}Q_3\\
	& = \frac{y\sigma^{2}(y)}{2}
	\Big( \vartheta'(y)\dx^2 \Ddx P_0 -  \psi'(y) \dx \Ddx P_1\Big) - \frac{\sigma^2(y)-\sigmaLim^2}{4}\psi(y)\Ddx^2 P_0 \\
	& \quad + \sqrt{\eps}\left[
	-\frac{\psi(y)}{2} \widetilde{\Omega} \dx \Ddx^2 P_0
	+ \frac{\sigma^2 (y)-\sigmaLim^2}{4} \vartheta(y)\dx\Ddx^2 P_0
	- \frac{\sigma^2 (y)-\sigmaLim^2}{4} \psi(y)\Ddx^2 P_1\right]\\
	& = \frac{y\sigma^{2}(y)}{2}\vartheta'(y) \left(\dx^4 - \dx^3 \right) P_0 -  \frac{\varpi(y)}{2} \left(\dx^3 - \dx^2 \right) P_1
	 -\frac{\sigma^2(y)-\sigmaLim^2}{4}\psi(y)\left(\dx^4 - 2  \dx^3 + \dx^2\right)P_0 \\
	& \quad + \frac{\sqrt{\eps}}{2}
	\Bigg[\left(\frac{\sigma^2 (y)-\sigmaLim^2}{2} \vartheta(y) - \psi(y) \widetilde{\Omega} \right) \left(\dx^5-2\dx^4 + \dx^3 \right) P_0
	 - \frac{\sigma^2 (y)-\sigmaLim^2}{2} \psi(y)\left(\dx^4 - 2  \dx^3 + \dx^2\right) P_1\Bigg].
\end{align*}
Exploiting the fact that we chose $P_1= -(T-t)\widetilde{\Omega}\left( \dx^3
- \dx^2\right)P_0$ (as in~\eqref{Eq_for_P1}), we obtain
\begin{align}\label{Eq_Feps_fin}
	F_\eps(t,x,y)
	& = -\frac{\sigma^2(y)-\sigmaLim^2}{4} \psi(y) \dx^2 P_0 + \left(
	- \frac{y\sigma^2(y)}{2}\vartheta'(y)
	+ \frac{\sigma^2(y)-\sigmaLim^2}{2}\psi(y)\right)\dx^3  P_0 \nonumber\\ 
	& \quad + \left(\frac{y\sigma^{2}(y)}{2} \vartheta'(y)
	 - \frac{\sigma^2(y)-\sigmaLim^2}{4}\psi(y)\right) \dx^4 P_0
	+\frac{T-t}{2}
	\varpi(y)\widetilde{\Omega} \left(\dx^4
	-2 \dx^5 + \dx^6\right)P_0 \nonumber\\ 
	& \quad +
	\frac{\sqrt{\eps}}{2}\Bigg\{ 
	\left[ \frac{\sigma^2(y)-\sigmaLim^2}{2}\vartheta(y)-\psi(y)\widetilde{\Omega} \right] \left(\dx^3 - 2\dx^4 + \dx^5\right)P_0 \nonumber\\
	& \quad + (T-t)\left( \frac{\sigma^2(y)-\sigmaLim^2}{2}\psi(y)\widetilde{\Omega} \right)\left(-\dx^4 + 3\dx^5 - 3\dx^6 + \dx^7\right)P_0
	\Bigg\}.
\end{align}
Performing similar computations for $G_\eps$, we obtain
\begin{align}\label{Eq_Geps_fin}
	G_\eps(x,y)
	& = Q_2(T,x,y)+\sqrt{\eps}Q_3(T,x,y)\nonumber\\ 
	& = -\frac{\psi(y)}{2}\Ddx P_0
	+ \sqrt{\eps}\left(\frac{\vartheta(y)}{2}\dx\Ddx P_0
	 - \frac{\psi(y)}{2}\Ddx P_1\right)\nonumber\\
	& = -\frac{\psi(y)}{2}\left( \dx^2-\dx\right)P_0
	+\frac{\sqrt{\eps}}{2}\bigg[ \vartheta(y) 
	\left( \dx^3-\dx^2\right)
	 +\psi(y) (T-T)\widetilde{\Omega}\left( \dx^5
	- 2\dx^4
	+\dx^3 \right)\bigg]P_0\nonumber\\
	& = -\frac{\psi(y)}{2}\left( \dx^2
	 - \dx\right)P_0
	+ \frac{\sqrt{\eps}}{2} \vartheta(y)\left( \dx^3 - \dx^2\right)P_0,
\end{align}
with $P_0, P_1$ evaluated at $(T,x)$.
The probabilistic representation of $Z^\eps$ as the solution of the Poisson equation in~\eqref{Eq_PDE_Z^vareps} reads
\begin{equation}\label{Eq_Zeps_prob_repres}
	Z^\eps(t,x,y)= \eps \ \EE_{t,x,y}\left[G_\eps(X_T,Y_T) + \int_t^T F_\eps(s,X_s,Y_s) \D s\right].
\end{equation}
To show that this is of order $\Oo(\eps)$ as $\eps\downarrow 0$, it is enough to bound~$F_\eps$ and~$G_\eps$ uniformly in~$\eps$.
The key ingredients here are the following two lemmas.  The proof of the first one, being long and technical, is postponed to Appendix~\ref{proof:pol_growth_poi_eq}.

\begin{lemma}\label{lem_pol_growth_poi_eq}
Let~$\xi$ be a solution to the Poisson equation 
$\LY \xi = g$
on $(y_\sigma, \infty)$,
with
\begin{equation*}
\left\{
\begin{array}{rll}
|g(y)| & \leq C, & \text{ for }y \in (y_\sigma, \underline y),\\
|g(y)| & \leq C \left(1+|y|^n\right), & \text{ for }y \geq \overline y,\\
\langle g, \Pi \rangle & = 0,
\end{array}
\right.
\end{equation*}
for some $C>0$, $n\in\mathbb{N}$, 
$\underline{y} \in (y_\sigma, \overline{y})$.
Then there exist $C'>0$, $n'\in\mathbb{N}$, 
$y_\sigma < \underline{y}' < \overline{y}'$ such that
\begin{equation}
\left\{
\begin{array}{rll}
|\xi'(y)| & \leq C', 
& \text{ for }y \in (y_\sigma, \underline y'),\\
|\xi'(y)| & \leq C'(1+|y|^{n'}),
& \text{ for }y \geq \overline y,
\end{array}
\right.
\end{equation}
and consequently
\begin{equation}
\left\{
\begin{array}{rll}
|\xi(y)| & \leq C'', 
& \text{ for }y \in (y_\sigma, \underline y'),\\
|\xi(y)| & \leq C''(1+|y|^{n'+1}),
& \text{ for }y \geq \overline y,
\end{array}
\right.
\end{equation}
with $C''$ suitable positive constant.
\end{lemma}

\begin{lemma}\label{lem_bound_der_price}
	If~$h$ is smooth and bounded with bounded derivatives, then
	$\partial_{x}^n P_0$ exists and is bounded
	for any $n \in \mathbb{N}$.
\end{lemma}

\begin{proof}
Since $P_0(t,x)$ is the BS price with constant volatility~$\sigmaLim^2$, denoting~$f(\cdot)$ the density function of $\mathcal{N}\left(-\half\sigmaLim^2(T-t),\sigmaLim^2\sqrt{T-t}\right)$ and assuming that the first $n$ derivatives of the function $h$ are uniformly bounded by $K >0$, we have, for $n=0$,
\begin{align*}
    |P_0(t,x)| 
    = \left|\int_{\RR} h(\E^{x+z})f(z)\D z\right|
    \leq \int_{\RR} \left|h(\E^{x+z})\right|f(z)\D z
    \leq K
\end{align*}
and then, for any $n \geq 1$,
\begin{align*}	
\dx^n P_0(t,x)
& = \dx^n \left(\int_{\RR} h(\E^{x+z})f(z)\D z\right)
= \dx^{n-1} \left(\int_{\RR} h'(\E^{x+z})\E^{x+z}f(z)\D z\right)\\
& = \dx^{n-2} \left(\int_{\RR} (h''(\E^{x+z})\E^{2(x+z)}+h'(\E^{x+z})\E^{x+z})f(z)\D z\right)\\
& = \cdots
= \int_{\RR} \sum_{k=1}^n \binom{n}{k} \dx^{k} h(\E^{x+z})\E^{k(x+z)}f(z)\D z \\
&= \int_{\RR} \sum_{k=1}^n \binom{n}{k} \dx^{k} h(\E^{x+z})\E^{k(x+z)}f(z)\D z,
\end{align*} 
and so
\begin{align*}	
|\dx^n P_0(t,x)|
& \leq   \sum_{k=1}^n \binom{n}{k} \int_{\RR} \left|\dx^{k} h(\E^{x+z})\right|\E^{k(x+z)}f(z)\D z
\leq K \int_{\RR} \sum_{k=1}^n \binom{n}{k}\E^{k(x+z)}f(z)\D z \\
& = K  \sum_{k=1}^n \binom{n}{k}\E^{kx}\E^{-k\half\sigmaLim^2(T-t)}\E^{k^2\half\sigmaLim^2(T-t)}
= K  \sum_{k=1}^n \binom{n}{k}\E^{kx+(k^2-k)\half\sigmaLim^2(T-t)}.
\end{align*} 
Since this is clearly finite, the lemma follows.
\end{proof}

Now, $\psi$ and $\vartheta$ are respectively the solutions to the Poisson equations~\eqref{Eq_eq_def_psi}-\eqref{Eq_eq_vartheta_def} and satisfy the hypotheses in Lemma~\ref{lem_pol_growth_poi_eq}.
Indeed, for $\psi$, the function $g:y\mapsto \sigma^2(y)-\sigmaLim^2$ 
clearly satisfies $\langle g, \Pi \rangle = \langle \sigma^2(\cdot)-\sigmaLim^2, \Pi \rangle =  \langle \sigma^2(\cdot), \Pi \rangle-\sigmaLim^2= \sigmaLim^2-\sigmaLim^2=0$. 
Furthermore, on $(y_\sigma, \infty)$,
\begin{align*}
	|g(y)|
	& = |\sigma^2(y)-\sigmaLim^2|
	\leq \sigma^2(y)+ \sigmaLim^2
	= \left(\frac{\alpha}{\beta}- \gamma y^{-\beta}\right)^2 +\sigmaLim^2
	\leq \frac{\alpha^2}{\beta^2}+\sigmaLim^2
\end{align*}
is finite.
Analogously, for $\vartheta$, the function $g:y\mapsto \varpi(y)- \langle \varpi, \Pi \rangle$
clearly satisfies
$\langle g, \Pi \rangle
	= \langle \varpi(y)- \langle \varpi, \Pi \rangle, \Pi \rangle
	= \langle \varpi, \Pi \rangle - \langle  \varpi, \Pi \rangle
	= 0$.
Clearly, $\langle \varpi,\Pi\rangle$ is a finite positive constant. 
Let us check the polynomial growth assumption on $(y_\sigma, \infty)$. 
Since~$\sigma$ is bounded there and since~$\psi$ (and its first derivative) has polynomial growth,
then
\begin{align*}
	|g(y)|
	& =|\varpi(y)- \langle  \varpi, \Pi \rangle|
	=|\varpi(y)| + \langle \varpi,\Pi\rangle
	\leq |y| |\sigma^2 (y)| |\psi'(y)| + \langle \varpi,\Pi\rangle\\
	& \leq |y| \frac{\alpha^2}{\beta^2}K' \left(1+|y|^{n'}\right) + \langle \varpi,\Pi\rangle
	\leq \left(\frac{\alpha^2}{\beta^2}K' +\langle \varpi,\Pi\rangle\right)\left(1+|y|^{n'+1}\right),
\end{align*}
which yields the desired growth condition.
Thus, $\psi$ and $\vartheta$ have at most polynomial growth at infinity, which we denote~$n_\psi$ and~$n_\vartheta$, and are bounded by a suitable constant when approaching~$y_\sigma$.
Plugging~\eqref{Eq_Feps_fin} and~\eqref{Eq_Geps_fin} in~\eqref{Eq_Zeps_prob_repres}, we can write
\begin{align*}
	Z^\eps(t,x,y)
	& = \eps \EE_{t,x,y}\left[G_\eps(X_T,Y_T) + \int_t^T F_\eps(s,X_s,Y_s) \D s\right]\\
	& = \eps \EE_{t,x,y}\Bigg[
	-\half \psi(Y_T)(\partial_x^2-\partial_x) P_0(T,X_T) 
	+ \frac{\sqrt{\eps}}{2} \vartheta(Y_T)(\partial_x^3-\partial_x^2)P_0(T,X_T) \\
	& + \int_t^T 
	\Bigg\{-\frac{\sigma^2(Y_s) - \sigmaLim^2}{4} \psi(Y_s)\dx^2
	 + \left( -\frac{\sigma^2(Y_s)}{2}Y_s\vartheta'(Y_s)
	 + \frac{\sigma^2(Y_s)-\sigmaLim^2}{2}\psi(Y_s)\right)\dx^3  \\ \nonumber
	& + \left[\frac{\sigma^{2}(Y_s)}{2} Y_s\vartheta'(Y_s)-\frac{\sigma^2(Y_s)-\sigmaLim^2}{4}\psi(Y_s)\right] \dx^4
	+\frac{T-s}{2}Y_s\sigma^2(Y_s) \psi'(Y_s)\widetilde{\Omega} \left(\dx^4 - 2\dx^5+\dx^6\right) \\ \nonumber
	& +\sqrt{\eps}\Bigg\{ 
	\half\left[ \frac{\sigma^2(Y_s)-\sigmaLim^2}{2}\vartheta(Y_s)-\psi(Y_s)\widetilde{\Omega} \right] \left(\dx^3 - 2\dx^4 + \dx^5\right)\\ \nonumber
	& + (T-s)\frac{\sigma^2(Y_s)-\sigmaLim^2}{4}\psi(Y_s)\widetilde{\Omega} 
	\left(-\dx^4 + 3\dx^5 - 3\dx^6 + \dx^7\right)
	\Bigg\}P_0(s, X_s) \D  s \Bigg].
\end{align*}
Now, an application of Lemma~\ref{lem_bound_der_price} yields
(with $\zeta:=\frac{\alpha^2}{\beta^2}+\sigmaLim^2$)
\begin{align*}
|Z^\eps(t,x,y)|
	& \preceq \eps \EE_{t,x,y}\Bigg[
	|\psi(Y_T)|
	+ \sqrt{\eps} |\vartheta(Y_T)|\\
	& + \int_t^T \Bigg\{\frac{\zeta}{4}|\psi(Y_s)|
	+ \frac{\alpha^2}{\beta^2}|Y_s||\vartheta'(Y_s)| + |\psi(Y_s)|\frac{\zeta}{2}
	 + \frac{\alpha^2}{2\beta^2}|Y_s||\vartheta'(Y_s)| + \frac{\zeta}{4}|\psi(Y_s)|\\
	& +\frac{2\alpha^2 (T-s)}{\beta^2}
	|Y_s| |\psi'(Y_s)|\widetilde{\Omega} +\sqrt{\eps}\Bigg[
	 \zeta|\vartheta(Y_s)|
	 + 2|\psi(Y_s)|\widetilde{\Omega}
	+ 2(T-s)\zeta|\psi(Y_s)|\widetilde{\Omega}
	\Bigg]\Bigg\} \D s \Bigg]\\
	&  \preceq \eps \EE_{t,x,y}\Bigg[
	|\psi(Y_T)| + \sqrt{\eps}|\vartheta(Y_T)|
	 + \int_t^T \Bigg\{|\psi(Y_s)| + |Y_s| |\vartheta'(Y_s)| +(T-s)|Y_s||\psi'(Y_s)| \\
	& \quad  +\sqrt{\eps}\Big(|\vartheta(Y_s)|+|\psi(Y_s)|+(T-s)|\psi(Y_s)|\Big) \Bigg\}\D s \Bigg],
\end{align*} 
where~$\preceq$ means less than modulo multiplication by some strictly positive constant.
Finally, applying Lemma~\ref{lem_pol_growth_poi_eq}, we obtain
\begin{align*}
    \left|Z^\eps(t,x,y)\right|
    & \preceq \eps\  \EE_{t,x,y}\Big[ \left(1+|Y_T|^{n_\psi}\right) + \sqrt{\eps}\left(1+|Y_T|^{n_\vartheta}\right) \\
    & \quad  + \int_t^T \Big\{1+|Y_s|^{n_\psi}+  |Y_s|\left\{\left[1+|Y_s|^{n_\vartheta-1}\right]
+ (T-s)\left[1+|Y_s|^{n_\psi-1}\right]\right\}\\
    & \quad +\sqrt{\eps}\big\{1+|Y_s|^{n_\vartheta} + (1+ (T-s)) |Y_s|^{n_\psi} \big\}\Big\} \D s \Big]
 \preceq\eps.
\end{align*} 
The finiteness in the last line
is a consequence of Appendix~\ref{App:Proof_boundedness_second_moment_Y}
on the uniform finiteness of the moments of~$Y$,
and the proposition thus follows.


\subsection{Proof of Lemma~\ref{lem_pol_growth_poi_eq}}
\label{proof:pol_growth_poi_eq}
With the notations introduced in Section~\ref{sec:Boundary}, 
the third assumption on~$g$ can be rewritten as
	\begin{align*}
		0 = \langle g, \Pi \rangle
		= \int_{y_\sigma}^{\infty} g(y) \Pi(\D y)
		= \left( \int_{y_\sigma}^{\infty}\frac{\D \xi}{\widetilde \sigma^{2}(\xi)s(\xi)}\right)^{-1} \int_{y_\sigma}^{\infty} g(y) m(y) \D y,
	\end{align*}
	and therefore
	\begin{equation}\label{Eq_center_cond_g}
		\int_{y_\sigma}^{\infty} g(y) m(y) \D y=0.	
	\end{equation}
Recall that the equation $\LY \xi = g$ 
solved by~$\xi$ on  $(y_\sigma, \infty)$ is equivalent to
	\begin{equation}
		\half\frac{\D}{\D M}\left(\frac{\D}{\D S} \xi (y)	\right) = g(y).
	\end{equation}
Integrating both sides yields
	\begin{align}\label{eq_xi'}
		\xi'(y)
		& = 2 s(y)\int_{y_\sigma}^{y} g(z) m(z) \D z.
	\end{align}
	We first study the behaviour around $y_\sigma$.
	Consider $y \in (y_\sigma, \underline{y})$, for a sufficiently small $\underline{y}$. 
	Since the function~$g$ is bounded by assumption, then
		\begin{align*}
		|\xi'(y)|
		& = 2 C s(y)\int_{y_\sigma}^{y} m(z) \D z
		= 2 C s(y)M_{\Pi}(y_\sigma, y].
	\end{align*}
	In the proof of the boundary classification of the left boundary point $y_\sigma < 1$ for the domain $(y_\sigma, \infty)$, we have seen that~\eqref{eq:Expan_right_rho_m}
	\begin{align}
		s(y)=\exp\left(\frac{\tbrK}{y-y_\sigma} + \tbrK\tbrKab\right)(y - y_\sigma)^{-\tbrK\overline{\chi_1}}(1+ \mathcal{O}(y-y_\sigma)), 
		\qquad\text{for }y \in (y_\sigma, \underline{y}),
	\end{align}	
	with $\tbrK = \frac{2y^{2\beta +1}_\sigma (1- y_\sigma)}{h \beta^2 \gamma^2}$, positive constant, and 
$M_{\Pi}(y_\sigma,y]
 = \int_0^{y-y_\sigma}\frac{\D x}{\overline \sigma^{2}(x)\overline{s}(x)}$,
with
\begin{align}
	\overline{s}(x)=\exp\left(\frac{\tbrK}{x} + \tbrK\tbrKab\right)x^{-\tbrK\overline{\chi_1}}(1+ \Oo(x)), 
	\qquad\text{for }x \in (0,\underline{y}-y_\sigma),
	\end{align}
	and
$\overline \sigma(x)^{-2}= \frac{y_\sigma^{2+2\beta}}{\beta^2 \gamma^2 x^2}(1 + \Oo(x))$, 
for $x \in (y_\sigma, \underline{y})$.
Thus, exploiting these two expansions, the change of variables $x=(y-y_\sigma)z$ and the asymptotic expansion for integrals in~\cite[Chapter~3.3, pages~62 and~67]{Miller}, we obtain
	\begin{align*}
	M(y_\sigma,y]
	& = \int_0^{y-y_\sigma}\frac{\D x}{\overline \sigma^{2}(x)\overline{s}(x)}
	= \frac{y_\sigma^{2+2\beta}}{\beta^2 \gamma^2}\int_0^{y-y_\sigma}\exp\left\{-\frac{\tbrK}{x} - \tbrK\tbrKab\right\}
	x^{\tbrK\overline{\chi_1}-2}(1+ \mathcal{O}(x)) \D x\\
	& = \frac{y_\sigma^{2+2\beta}}{\beta^2 \gamma^2}\E^{-\tbrK\tbrKab}(y-y_\sigma)^{\tbrK\overline{\chi_1}-1}\int_0^{1}\exp\left\{-\frac{\tbrK}{(y-y_\sigma)z}\right\}
	z^{\tbrK\overline{\chi_1}-2}(1+ \mathcal{O}(z)) \D z\\
	& = \frac{y_\sigma^{2+2\beta}}{\beta^2 \gamma^2}\E^{-\tbrK\tbrKab}(y-y_\sigma)^{\tbrK\overline{\chi_1}-1}\exp\left\{-\frac{\tbrK}{y-y_\sigma}\right\} \left(\frac{y-y_\sigma}{\tbrK} +  \mathcal{O}((y-y_\sigma)^2)\right)\\
	& = \frac{y_\sigma^{2+2\beta}}{\beta^2 \gamma^2 \tbrK}\E^{-\tbrK\tbrKab}(y-y_\sigma)^{\tbrK\overline{\chi_1}}\exp\left\{-\frac{\tbrK}{y-y_\sigma}\right\}
	\left(1+ \mathcal{O}(y-y_\sigma)\right).
	\end{align*}
	Thus, we conclude
	\begin{align}
	|\xi'(y)| \leq 2Cs(y)M(y_\sigma,y]
	= \frac{2y_\sigma^{2+2\beta}}{\beta^2\gamma^2} (1+ \mathcal{O}(y-y_\sigma)),
	\end{align}
	which yields the boundedness of $\xi'(y)$ and of  $\xi(y)$ itself as $y$ approaches $y_\sigma$.

About the behaviour at infinity,
applying the centering condition~\eqref{Eq_center_cond_g} to~\eqref{eq_xi'} yields
		\begin{align*}
		\xi'(y)
		& = 2 s(y)\int_{y_\sigma}^{y} g(z) m(z) \D z\\
		& = 2 s(y)\left(\int_{y_\sigma}^{y} g(z) m(z) \D z+\int_{y}^{\infty} g(z) m(z) \D z-\int_{y}^{\infty} g(z) m(z) \D z \right)\\
		& = - 2 s(y)\int_{y}^{\infty} g(z) m(z) \D z 
	\end{align*}
Since~$s$ and~$m$ are non-negative, the polynomial growth assumption in the statement of Lemma~\ref{lem_pol_growth_poi_eq} and the definition of~$m$ give
	\begin{align}\label{Eq_bound_xi'_1}
		|\xi'(y)|
		& = 2 s(y)\left|\int_{y}^{\infty} g(z) m(z) \D z\right|
		\leq 2 s(y)\int_{y}^{\infty} |g(z)| m(z) \D z \\ \nonumber
& \leq 2C s(y)\int_{y}^{\infty}  z^{n} m(z) \D z
 \leq 2C s(y)\int_{y}^{\infty}  \frac{z^{n-2}}{\sigma^2(z)s(z)} \D z. 
	\end{align}
Since~$\overline{y}$ can be picked as $\overline{y}> 1$, then $|z|^{n-2} \leq 1$ for $n \in \{0,1\}$,
and we thus take~$1$ in place of $z^{n-2}$.
	We make a short digression to study~$s(y)$, for $y \in (a, \infty)$ with $a >y_\sigma$. By definition,
	\begin{equation}
		s(y)= \exp\left\{-\int_a^y \frac{2b(\eta)}{\widetilde{\sigma}^2(\eta)}\D \eta\right\}
		= \E^{-f_a(y)},
	\end{equation}
with $f_a(y):=\int_a^y \frac{2b(\eta)}{\widetilde{\sigma}^2(\eta)}\D \eta$,
which we can compute explicitly as
$$
f_a(y)
= \frac{2}{h}\left(\int_a^y
\frac{\D\eta}{ \eta \left(-\frac{\alpha}{\beta}+\gamma \eta^{-\beta}\right)^2}
- \int_a^y \frac{\D\eta}{ \left(-\frac{\alpha}{\beta}+\gamma \eta^{-\beta}\right)^2} \right)
= \frac{2}{h}\Big(I_1(a,y) - I_2(a,y)\Big),
$$
where
\begin{align*}
I_1(a,y) & := \frac{\beta}{\alpha^2} \log \left( \frac{\beta \gamma - \alpha y^\beta}{\beta \gamma - \alpha a^\beta} \right) + \frac{\beta^2 \gamma}{\alpha } \frac{a^{-\beta}- y^{-\beta}}{(\alpha- \beta \gamma a^{-\beta})(\alpha- \beta \gamma y^{-\beta})},\\
I_2(a,y) & := \frac{1}{\gamma (2\beta+1)}
\left[y^{2\beta+1}{}_2 F_1 \left(2, 2+\frac{1}{\beta}; 3+\frac{1}{\beta}; \frac{\alpha y^{\beta}}{\beta \gamma}\right) - a^{2\beta+1}{}_2 F_1 \left(2, 2+\frac{1}{\beta}; 3+\frac{1}{\beta}; \frac{\alpha a^{\beta}}{\beta \gamma}\right) \right].
\end{align*}
\begin{itemize}
\item[-]  Since $y>a$, then
the first term in~$I_1$ satisfies $\frac{\beta \gamma - \alpha y^\beta}{\beta \gamma - \alpha a^\beta} \in (0,1]$ 
so that its logarithm is well-posed and negative.
		
\item[-]  Likewise, the second term in $I_1$ is positive 
and (as a function of~$y$) increasing and bounded by its $\infty$-limit equal to$\frac{a^{-\beta}}{(\alpha- \beta \gamma a^{-\beta})\alpha}$.
		
\item[-]  The two terms in $I_2$ can be rewritten exploiting the following series representation of the hypergeometric function \cite[Volume~I, Chapter~III, Section~3.6, Equation~(1)]{L69}, 
 which holds for any $|z| > 1$ and $a-b \notin \mathbb{Z}$:
\begin{align*}
    {}_2 F_1 \left(a,b;c;z\right)
    & = \frac{\Gamma(b-a)\Gamma(c)}{\Gamma(b)\Gamma(c-a)}\frac{1}{(-z)^{a}} \sum_{k=0}^{\infty}\frac{(a)_k(a-c+1)_k}{k!(a-b+1)_k}\frac{1}{z^{k}}\\
    & + \frac{\Gamma(a-b)\Gamma(c)}{\Gamma(a)\Gamma(c-b)}\frac{1}{(-z)^{b}} \sum_{k=0}^{\infty}\frac{(b)_k(b-c+1)_k}{k!(b-a+1)_k}\frac{1}{z^{k}}.
\end{align*}
In our specific case this reads
\begin{equation}
{}_2 F_1 \left(2, 2+\frac{1}{\beta}; 3+\frac{1}{\beta}; z\right)
= (2\beta+1)\sum_{k=0}^{\infty}\frac{k+1}{1-k \beta}z^{-(2+k)} +\half\Gamma\left(-\frac{1}{\beta}\right)
\Gamma\left(3+\frac{1}{\beta}\right) \left(-\frac{1}{z}\right)^{2+\frac{1}{\beta}},
\end{equation}
which implies
$-I_2(a,y) 
= \sum_{k\geq 0}\gimel_k
(a^{1-\beta k}-y^{1-\beta k})$,
where we define $\gimel_k:=\frac{k+1}{1-k \beta}\gamma^{k+1}\left( \frac{\alpha}{\beta}\right)^{k+2}$
for convenience.
We further introduce the useful quantities
$$
\Sigund_{n}^{z}
:= \sum_{k=0}^{n-1}\gimel_k
z^{1-\beta k}
\qquad\text{and}\qquad
\Sigov_{n}^{z}
:= \sum_{k=n}^{\infty}\gimel_k
z^{1-\beta k}.
$$
Then, for any $\beta \in (0,\half)$, there exists $n_\beta \in \mathbb{N}\setminus \{ 0,1,2\}$ such that $1-\beta n < 0$, for $n \geq n_\beta$, and $1-\beta n \geq 0$, for $n < n_\beta$. Hence, for any $z \in (y, \infty)$, with $y > a> y_\sigma$,
$$
-I_2(a,z) 
= C_{a} - \sum_{k=0}^{\infty}\gimel_k z^{1-\beta k}
 = C_{a} - \Sigund_{n_\beta}^{z} - \Sigov_{n_\beta}^{z}
 \leq C_{a} - 
 \Sigund_{n_\beta}^{z} - \Sigov_{n_\beta}^{y},
$$
where the constant $C_{a}:=\Sigund_{\infty}^{a}$
is finite.
\end{itemize}
As a consequence of these bullet points, we deduce
\begin{align*}
\frac{1}{s(z)} 
& = \E^{f_a(z)}
= \left( \frac{\beta \gamma - \alpha z^\beta}{\beta \gamma - \alpha a^\beta} \right)^{\frac{2\beta}{h\alpha^2}}
\exp\left\{\frac{2\beta^2 \gamma}{h\alpha } \frac{a^{-\beta}- z^{-\beta}}{(\alpha- \beta \gamma a^{-\beta})(\alpha- \beta \gamma z^{-\beta})} \right\}
\E^{-\frac{2}{h}I_2(a,z)}\\
& \leq \exp\left\{\frac{2\beta^2 \gamma}{h\alpha } \frac{a^{-\beta}}{(\alpha- \beta \gamma a^{-\beta})\alpha}\right\}
\exp\left\{\frac{2}{h}\left[
C_{a} - \Sigund_{n_\beta}^{z}- \Sigov_{n_\beta}^{y}\right]\right\}\\
& = \exp\left\{\frac{2\beta^2 \gamma a^{-\beta}}{(\alpha- \beta \gamma a^{-\beta})\alpha^2 h}+\frac{2 C_{a}}{h}\right\}
\exp\left\{- \frac{2}{h}\left(\Sigund_{n_\beta}^{z} + \Sigov_{n_\beta}^{y}\right)\right\}.
\end{align*}
Let us now go back to the starting problem and consider $y > a> y_\sigma$. 
Replacing the expression in~\eqref{Eq_bound_xi'_1}, we have
	\begin{align*}
		|\xi'(y)|
		& \leq 2 C s(y)\int_{y}^{\infty}  \frac{z^{n-2}}{\left(\frac{\alpha}{\beta}-\gamma z^{-\beta}\right)^2s(z)} \D z
		\leq 2 C \left(\frac{\alpha}{\beta}-\frac{\gamma}{ y^{\beta}}\right)^{-2} s(y)\int_{y}^{\infty}\frac{z^{n-2}}{s(z)} \D z\\
		& = 2 C \left(\frac{\alpha}{\beta}-\frac{\gamma}{ y^{\beta}}\right)^{-2} 
		\left( \frac{\beta \gamma - \alpha a^\beta}{\beta \gamma - \alpha y^\beta} \right)^{\frac{2\beta}{h\alpha^2} }  
		\E^{ \frac{\beta^2 \gamma}{\alpha } \frac{ y^{-\beta}-a^{-\beta}}{(\alpha- \beta \gamma a^{-\beta})(\alpha- \beta \gamma y^{-\beta})}}
		\E^{ \frac{2}{h}I_2(a,y)}\int_{y}^{\infty}\frac{z^{n-2}}{s(z)} \D z\\
		& \leq 2 C \left(\frac{\alpha}{\beta}-\frac{\gamma}{ y^{\beta}}\right)^{-2} \left( \frac{\beta \gamma - \alpha a^\beta}{\beta \gamma - \alpha y^\beta} \right)^{\frac{2\beta}{\alpha^2h}}
		\exp\left\{ -\frac{2}{h}\left[ C_{a}- \Sigund_{n_\beta}^{y}
- \Sigund_{n_\beta}^{y}\right]\right\} \times\\
		& \quad \times  \int_{y}^{\infty} z^{n-2}\E^{\frac{2\beta^2 \gamma a^{-\beta}}{(\alpha- \beta \gamma a^{-\beta})\alpha^2 h}+\frac{2}{h}C_{a}}
		\exp\left\{- \frac{2}{h}\Sigund_{n_\beta}^{z}- \frac{2}{h}\Sigov_{n_\beta}^{y}\right\} \D z\\
		& \leq 2 C \E^{\frac{2\beta^2 \gamma a^{-\beta}}{(\alpha- \beta \gamma a^{-\beta})\alpha^2 h }}
		\left(\frac{\alpha}{\beta}-\frac{\gamma}{ y^{\beta}}\right)^{-2} 
		\left( \frac{\beta \gamma - \alpha a^\beta}{\beta \gamma - \alpha y^\beta} \right)^{\frac{2\beta}{h \alpha^2} } 
		\exp\left\{\frac{2}{h}\Sigund_{n_\beta}^{y}\right\}
		\int_{y}^{\infty} z^{n-2}
		\exp\left\{- \frac{2}{h}
		\Sigund_{n_\beta}^{z}\right\} \D z.
	\end{align*}
	Now, suppose that the integral in the last line satisfies a bound of the form
\begin{equation}\label{Eq_int_ineq_exp}
\int_{y}^{\infty} z^{n-2}
\exp\left\{ - \frac{2}{h}
\Sigund_{n_\beta}^{z}\right\} \D z
	\leq K \sum_{j=0}^{N}\frac{N!}{j!}
	\left(\frac{2}{h}\Sigund_{n_\beta}^{y}\right)^j
	\exp\left\{-\frac{2}{h}
	\Sigund_{n_\beta}^{y}\right\},
\end{equation}
for some $K>0$ and some integer $N>0$.
Plugging this in the equation above, we obtain
	\begin{align*}
		|\xi'(y)|
		& \leq 2 K C \exp\left\{\frac{2\beta^2 \gamma a^{-\beta}}{(\alpha- \beta \gamma a^{-\beta})h \alpha^2}\right\}
		\left(\frac{\alpha}{\beta}-\gamma y^{-\beta}\right)^{-2} 
		\left( \frac{\beta \gamma - \alpha a^\beta}{\beta \gamma - \alpha y^\beta} \right)^{\frac{2 \beta}{h \alpha^2} } 
		\sum_{j=0}^{N}\frac{N!}{j!}
		\left(\frac{2}{h}\Sigund_{n_\beta}^{y}\right)^j \\
		& =  2 K C \exp\left\{\frac{2\beta^2 \gamma a^{-\beta}}{(\alpha- \beta \gamma a^{-\beta})h \alpha^2}\right\}
		\frac{\left( \beta \gamma - \alpha a^\beta\right)^{\frac{2\beta}{h\alpha^2}}}{\beta^2}
		\frac{y^{2\beta}}{\left(\beta \gamma - \alpha y^\beta \right)^{2+\frac{2\beta}{\alpha^2 h}}} 
		\sum_{j=0}^{N}\frac{N!}{j!}
		\left(\frac{2}{h}\Sigund_{n_\beta}^{y}\right)^j\\
		& \leq  2 K C
		\exp\left\{\frac{2\beta^2 \gamma a^{-\beta}}{(\alpha- \beta \gamma a^{-\beta})h \alpha^2}\right\}
		\frac{\left( \beta \gamma - \alpha a^\beta\right)^{-2}}{\beta^2}
		y^{2\beta}
		\sum_{j=0}^{N}\frac{N!}{j!}
		\left(\frac{2}{h}\Sigund_{n_\beta}^{y}\right)^j \\
		& \leq K_{a} 
		\left(1+ |y|^{2\beta+N(1-\beta n_{\beta})}\right)
		= K_{a} 
		\left(1+ |y|^{N_{\beta}}\right),
	\end{align*}
	
where $N_\beta$ and $K_{a}$ are respectively a suitably chosen positive integer and a positive constant. 
Thus, this last inequality yields the desired polynomial growth for $\xi'$ and $\xi$.
Finally the inequality in~\eqref{Eq_int_ineq_exp} is a consequence of the following lemma.

\begin{lemma}
Let $y>a>1$, $k,N \in \mathbb{N}_0$ and
$I:= \int_y^{\infty} z^{k} \exp \{- \sum_{j=0}^N A_j z^{d_j} \} \D z$,
with $(A_j)_{j \in \{0, \dots, N\}} \geq 0$ and $(d_j)_{j \in \{0, \dots, N\}} \in (0,1)$.
Then, there exists $r\in\mathbb{N}$ and $C>0$ such that
$$
I \leq C \exp\left\{- \sum_{j=0}^N A_j y^{d_j} \right\} \sum_{k=0}^{r} \frac{r!}{k!} \left(\sum_{j=0}^N A_j y^{d_j} \right)^k.
$$
\end{lemma}

\begin{proof}
The function $g: (y, \infty) \to \RR_{+}$, defined as $g(z):= \sum_{j=0}^N A_j z^{d_j}$, is positive and strictly increasing, hence invertible.
Its inverse~$g^{\leftarrow}$ is thus strictly increasing and
$\lim_{z \uparrow\infty} g(z)=+\infty$.
The change of variables $g(z)=u$ thus implies
$$
I=\int_{g(y)}^{\infty} \E^{-u} \frac{g^{\leftarrow}(u)^{k}}{g'\left(g^{\leftarrow}(y)\right)} \D u.
$$
Notice that the first derivative of $g$, given by $g'(z)= \sum_{j=0}^N A_j d_j z^{d_j-1}$, is clearly positive and strictly decreasing on $(y,\infty)$.
Now, set $\amin=\min_{j \in \{0, \dots, N\}} A_j$ and $\dmin=\min_{j \in \{0, \dots, N\}} d_j$. 		
Since
$g(z)\geq \amin z^{\dmin}$,
then
$g\left(\left(\frac{z}{\amin}\right)^{\frac{1}{\dmin}}\right)\geq z$,
and so, by the monotonicity of $g^{\leftarrow}$, we have
\begin{align}\label{Eq_ineq_g^-1}
g^{\leftarrow}(z) 
\leq g^{\leftarrow}\left(g\left(\left(\frac{z}{\amin}\right)^{\frac{1}{\dmin}}\right)\right)
\leq \amin^{-\frac{1}{\dmin}}z^{\frac{1}{\dmin}},
\end{align}
as well as $g'(z) \geq \amin\dmin z^{\dmin-1}$.
Applying this inequality and then~\eqref{Eq_ineq_g^-1} to the chain of inequalities for~$I$, gives, for a suitably chosen positive integer~$r$ some constant $C>0$,
\begin{align*}
I 
& \leq \int_{g(y)}^{\infty} \E^{-u} \frac{g^{\leftarrow}(u)^{k}}{\amin\dmin g^{\leftarrow}(u)^{\dmin-1}} \D u
= \frac{1}{\amin\dmin} \int_{g(y)}^{\infty} \E^{-u}g^{\leftarrow}(u)^{k+1-\dmin} \D u	\\	
& \leq \frac{1}{\amin\dmin} \int_{g(y)}^{\infty} \E^{-u}\left(\frac{u}{\amin}\right)^{\frac{1}{\dmin}(k+1-\dmin)} \D u	
= \amin^{-\frac{k+1}{\dmin}} \dmin^{-1}  \int_{g(y)}^{\infty} \E^{-u}u^{\frac{1}{\dmin}(k+1-\dmin)} \D u	\\
& \leq C \int_{g(y)}^{\infty} \E^{-u}u^{r} \D u
= \E^{-g(y)}
\sum_{j=0}^r \frac{r!}{j!}g(y)^{j},	
\end{align*}
which ends the proof of the inequality in the statement of the theorem.		
\end{proof}



\subsection{Uniform bounds for the moments of~$Y$}
\label{App:Proof_boundedness_second_moment_Y}


    Because of Section~\ref{sec:existence_uniqueness_sol}, Theorem~\ref{thm:ErgodicB} and Proposition~\ref{prop:Stationary}, we restrict our interest to the case with $\beta \in (0,\half)\cup \{1\}$ and domain $\Dom=(y_\sigma, \infty)$, with $y_\sigma:= \left( \frac{\beta\gamma}{\alpha}\right)^{1/\beta} < 1$.
	We need to prove that, for any $n\in \NN$, the uniform (in time) bound
	$\sup_{t \geq 0} \EE[Y^n_t] \leq K$ holds.
	We shall use the following lemma, the proof of which is relegated below:
	
\begin{lemma}\label{lem:Claim}
On any compact interval of the form $[0,T]$,
 any moment of~$Y$ is uniformly bounded
 and $\lim_{t \to s}	 \EE[(Y_t-Y_s)^{n+2}] = 0$.
\end{lemma}

This claim implies immediately that  $\EE[Y_t^{n+2}]$, $\EE[Y_t^{n+1}]$, $\EE[Y_t^n]$, $\EE[Y_t^{n-\beta}]$ and $\EE[Y_t^{n-2\beta}]$ are all continuous on any compact interval.
Now, It\^o's formula implies yields
	\begin{align*}
		Y_t^n 
		& = y_0^n + \int_0^t nY_s^{n-1} \D Y_s +\half \int_0^t n(n-1)Y_s^{n-2} \D \langle Y\rangle^2_s \\
		& = y_0^n + \frac{n}{h}\int_0^t 
		\left(Y_s^n -Y_s^{n+1}\right) \D s + \int_0^t \left(-\frac{\alpha n}{\beta}Y_s^n + \gamma n Y_s^{n-\beta} \right)\D W_s\\
		& \quad +\frac{n(n-1)}{2}\int_0^t \left(\frac{\alpha^2}{\beta^2}Y_s^n +\gamma^2 Y_s^{n-2\beta} - \frac{2\alpha\gamma}{\beta}Y_s^{n-\beta}\right)\D s.
	\end{align*}
	Taking expectations on both sides and exploiting the regularity of the processes involved (from the aforementioned claim) we obtain
	\begin{align*}
		\EE[Y_t^n] 
		& = y_0^n + \left(\frac{n}{h}+\frac{\alpha^2n(n-1)}{2\beta^2}\right)\int_0^t \EE[Y_s^n] \D s - \frac{n}{h} \int_0^t \EE[Y_s^{n+1}] \D s + 0 +\frac{\gamma^2n(n-1)}{2} \int_0^t \EE[Y_s^{n-2\beta}]\D s\\
		& \quad -\frac{\alpha\gamma n(n-1)}{\beta}\int_0^t \EE[Y_s^{n-\beta}]\D s.
	\end{align*}
	Define now the function $t\mapsto\varphi(t):=\EE[Y_t^n]$,
	which is differentiable since on any compact $[0,T]$, $|\partial_{t}f(t,Y)|$ is bounded in~$L^1$, for $f(t,Y) := \int_{0}^{t}Y_s^n\D s$.
	Since the process~$Y$ is positive almost surely, differentiating the expression above and applying H\"{o}lder inequality yield
	\begin{align*}
		\varphi'(t) 
		& =  \left(\frac{2}{h}-\frac{\alpha^2}{\beta^2}\right)\varphi(t) - \frac{2}{h} \EE\left[Y_t^3\right]
		+\gamma^2 \EE\left[Y_t^{2-2\beta}\right] - \frac{2\alpha\gamma}{\beta}\EE\left[Y_t^{2-\beta}\right]\\
		& \leq  \left(\frac{n}{h}+\frac{\alpha^2n(n-1)}{2\beta^2}\right)\varphi(t) - \frac{n}{h} \EE\left[Y_t^{n+1}\right] +\frac{\gamma^2n(n-1)}{2} \EE\left[Y_t^{n-2\beta}\right] \\
		& \leq  \left(\frac{n}{h}+\frac{\alpha^2n(n-1)}{2\beta^2}\right)\varphi(t) - \frac{n}{h} \varphi(t)^{1+ \frac{1}{n}} +\frac{\gamma^2 n(n-1)}{2} \varphi(t)^{1-\frac{2}{n}\beta}
		= \psi(\varphi(t)),
	\end{align*}
	with $\psi(y):= \left(\frac{n}{h}+\frac{\alpha^2n(n-1)}{2\beta^2}\right)y - \frac{n}{h} y^{1+ \frac{1}{n}} +\frac{\gamma^2 n(n-1)}{2} y^{1-\frac{2}{n}\beta}$.
	Since $\lim_{y\uparrow\infty}\psi(y)=-\infty$, there exists~$y^{*}$ such that $\psi(y) \leq -1$ for all $y \geq y^*$.
	
	This implies that $\varphi(\cdot)$ is uniformly bounded.
	First, without loss of generality we can assume $y^* \geq y_0^2$.
	Now, either the level $y^*$ is never reached, so that that the function~$\varphi$ is uniformly bounded by~$y^*$, 
	or that~$y^*$ is actually attained at some time~$t^*$, namely $\varphi(t^*)=y^*$.
	Let us show that in this last case the level $y^*+1$ cannot be attained and consequently~$\varphi$ is uniformly bounded by $y^*+1$.
	Assume by contradiction that there exists $\overline{t}$ such that $\varphi(\overline{t})=y^*+1$.
	Since~$\varphi$ is continuous,
	then $\overline{t}\geq t^*$.
	Set $\widehat{t}:=\max\{ 0 \leq t \leq \overline{t}: \varphi(t)=y^* \}$.
	Clearly then $\psi(\varphi(t)) \leq -1$ for all $t \in [\widehat{t}, \overline{t}]$,
	and furthermore
$$
y^*+1 
		=\varphi(\overline{t})
		= \varphi(\widehat{t})+ \int^{\overline{t}}_{\widehat{t}} \varphi'(t) \D t
		= y^* + \int^{\overline{t}}_{\widehat{t}} \varphi'(t) \D t
		\leq y^* + \int^{\overline{t}}_{\widehat{t}} \varphi'(t) \D t
		 \leq y^* + \int^{\overline{t}}_{\widehat{t}} \psi(\varphi(t)) \D t  \leq y^*,
$$
which is obviously a contradiction and thus completes the proof.
	\\	
	
We now prove Lemma~\ref{lem:Claim}.
\begin{proof}[Proof of Lemma~\ref{lem:Claim}]
The finiteness of any moment of~$Y$ can be recovered proceeding as in \cite{DET14}.
    Indeed, let $\tau_{M}:=\inf \{ t \geq 0: Y_t \geq M \}$ for any $M>0$, 
    so that $Y_{t  \land \tau_{M}} \leq M$ and hence is bounded almost surely. 
    Consider a function~$h \in \mathcal{C}^2([0,\infty))$ with the following properties:
    \begin{equation*}
    \left\{
    \begin{array}{ll}
        h(y) = 1, &  y \leq \half, \\
        h(y) \geq y^k, & \text{everywhere,}\\
        h(y) = y^k, & y \geq 2.
        \end{array}
        \right.
    \end{equation*}
    It is then easy to see that
    there exists a constant $\widetilde{C}>0$ such that, for all $y\geq 0$,
$$
\frac{\widetilde{\sigma}^2(y)}{2}h''(y)
 + b h'(y) \leq \widetilde{C} h(y).
$$
    Then, set  $f(t):= \EE_{y_0}[h(Y_{t \land \tau})]$. 
    It\^{o}'s formula implies
    \begin{align*}
        f(t) 
        & = h(y_0) + \EE_{y_0}\left[\int_0^{t \land \tau} \frac{\widetilde{\sigma}^2(Y_s)}{2}h''(Y_s)+b(Y_s)h'(Y_s) \D s \right]\\
        & = h(y_0) + \widetilde{C}\  \EE_{y_0}\left[\int_0^{t \land \tau} h(Y_s) \D s \right]
        = h(y_0) + \widetilde{C} \EE_{y_0}\left[\int_0^{t \land \tau} h(Y_{s\land \tau}) \D s \right]\\
        & \leq h(y_0) + \widetilde{C}\  \EE_{y_0}\left[\int_0^{t} h(Y_{s\land \tau}) \D s \right]
        = h(y_0) + \widetilde{C} \int_0^{t} f(s) \D s.
    \end{align*}
    Finally, an application of Gronwall's inequality yields
    \begin{align*}
        \EE_{y_0}\left[ Y_{t \land \tau}^k\right]
        & \leq \EE_{y_0}\left[ h(Y_{t \land \tau})\right] 
        \leq h(y_0)\E^{\widetilde{C}t}
        \leq C \left(1+y_0^k\right),
    \end{align*}
    which does not depend on~$M$, proving the uniform finiteness of moments of~$Y$ on~$[0,T]$.
    
Regarding the second item of the lemma, applying, in sequence, H\"{o}lder, BDG and H\"{o}lder inequalities, Fubini's Theorem and the previously boundedness of moments of~$Y$, we obtain
    \begin{align*}
    	\EE[(Y_t-Y_s)^n]
    	& = \EE\Bigg[ \sum_{k=0}^n \binom{n}{k}\left(\int_s^t b(Y_u) \D u \right)^{n-k}\left(\int_s^t \widetilde \sigma(Y_u) \D W_u\right)^k\Bigg]\\ \nonumber
    	& \leq \sum_{k=0}^n \binom{n}{k} \EE\Bigg[\left(\int_s^t b(Y_u) \D u \right)^n\Bigg]^\frac{n-k}{n} \EE\Bigg[\left(\int_s^t \widetilde \sigma(Y_u) \D W_u\right)^n\Bigg]^{\frac{k}{n}}\\ \nonumber
    	& \leq \sum_{k=0}^n \binom{n}{k} \EE\Bigg[\int_s^t b(Y_u)^n \D u \Bigg]^\frac{n-k}{n} \EE\Bigg[\left(\int_s^t \widetilde \sigma(Y_u)^2 \D u\right)^{\frac{n}{2}}\Bigg]^{\frac{k}{n}}\\ \nonumber
    	& \leq \sum_{k=0}^n \binom{n}{k} \EE\Bigg[\int_s^t b(Y_u)^n \D u \Bigg]^\frac{n-k}{n} \EE\Bigg[\int_s^t \widetilde \sigma(Y_u)^n \D u\Bigg]^{\frac{k}{n}}\\ \nonumber
    	& \leq \sum_{k=0}^n \binom{n}{k} \left\{\frac{1}{h^n}\int_s^t \EE[Y_u^n(1-Y_u)^n] \D u\right\}^{\frac{n-k}{n}}\left\{\int_s^t \EE\left[Y_u^n\left(-\frac{\alpha}{\beta}+\gamma Y_u^{-\beta}\right)^n\right] \D u\right\}^\frac{k}{n} \\ \nonumber
    	& \leq \sum_{k=0}^n \binom{n}{k} \frac{1}{h^{n-k}}\left\{\int_s^t \EE[Y_u^n+Y_u^{2n}] \D u\right\}^{\frac{n-k}{n}} 2^{\frac{k(n-1)}{n}}\left\{\int_s^t \EE\left[\frac{\alpha^n}{\beta^n}Y_u^n +\gamma^n Y_u^{n(1-\beta)}\right] \D u\right\}^\frac{k}{n}.
    \end{align*}
    Since all moments of~$Y$ are uniform bounded over $[0,T]$, we obtain
$$
\lim_{t \to s}\EE[(Y_t-Y_s)^n]
\leq 
\lim_{t \to s}C(T, y_0,n) (t-s) = 0,
$$
completing the proof.
\end{proof}

\section{Large deviations proofs}
\subsection{Proof of Proposition~\ref{pp:LDP_Y}}\label{proof:LDP_Y}
Since the process~$Y^\eps$ lies in~$\RR^*_+$ instead of~$\RR$, 
we adapt the proof of~\cite[Theorem 2.9]{P07} to prove a large deviations principle with speed~$\eps$ and rate function~$\Ir^Y$. 
Since $y_\sigma >0$, and in both cases $y_0 \ge y_\sigma$ and $y_0 < y_\sigma$, 
the function~$\widetilde{\sigma}$ is locally Lipschitz continuous on~$\RR^*_+$.
Furthermore, for $f\in\overline{\Hh}$, the Picard-Lindel\"of Theorem implies that the controlled ODE 
$\dot{g}_t = \widetilde{\sigma} (g_t) \dot{f}_t$, with $g_0 = y_0$ admits the solution
$$
\Ss_2^{y_0} (f)(t) = \left( \frac{\beta\gamma}{\alpha}\right)^{\frac{1}{\beta}} \left[ \E^{-\alpha \int_0^t \dot{f}_u \D u} \left( y_0^\beta \frac{\alpha}{\beta\gamma}-1\right)+1\right]^{1/\beta},
\qquad\text{for }t\in[0,T],\quad y_0>0.
$$
This formulation requires the term $\left[ \E^{-\alpha \int_0^t \dot{f}_u \D u} \left( y_0^\beta \frac{\alpha}{\beta\gamma}-1\right)+1\right]$ to be positive for all $y_0>0$:
\begin{itemize}
\item[-]  If $y_0 \ge y_\sigma$, then $y^\beta_0 \frac{\alpha}{\beta\gamma} -1 \ge 0$ and $\Ss_2^{y_0} (f)$ is positive on~$[0,T]$;
\item[-]  If $y_0 < y_\sigma$, then $y^\beta_0 \frac{\alpha}{\beta\gamma} -1 <0$ and $\Ss_2^{y_0} (f)$ is positive on~$[0,T]$ if and only if~\eqref{eq:Condition_f} holds.
\end{itemize}
The crucial step in~\cite[Theorem 2.9]{P07} is~\cite[Theorem 2.7]{P07}, which states that if $\sqrt{\eps}W$ is close to $f\in\overline{\Hh}$, then $Y^\eps$ should be close to $\Ss_2^{y_0}(f)$, the solution of the controlled ODE.
The case of bounded and locally Lipschitz coefficients on the whole real line was done in~\cite[Theorem~2.7]{P07},
but with such conditions on a domain,
a new localisation argument is required.
Given suitable $\eta> \delta >0$, with~$\delta$ sufficiently small, there exists $r \in (0, \eta)$ such that the $\delta$-tube around $\Ss_2^{y_0}(f)$ is contained in $B_r (\eta)$. 
For this radius~$r$ to exist, one simply needs to make sure that the solution~$\Ss_2^{y_0}(f)$ of the controlled ODE never reaches zero (explosion is impossible as infinity is recurrent), which is obvious 
when $y_0 \ge y_\sigma$, and guaranteed by Condition~\eqref{eq:Condition_f} when $y_0 < y_\sigma$.
Then both functions
$$
\mathfrak{b} (x) :=
\left\{
\begin{array}{ll}\displaystyle
b(x), & x \in [\eta-r,\eta+r],
\\
\displaystyle 
b\left(\frac{(\eta-r)x}{|x|} \right)
 = b(\eta-r) , & x  < \eta-r,\\
\displaystyle 
b\left(\frac{(\eta+r)x}{|x|} \right)
 = b(\eta+r) , & x  > \eta+r,
\end{array}
\right.
$$
and
$$
\mathfrak{s}(x) :=
\left\{
\begin{array}{lll}
\displaystyle \widetilde{\sigma} (x), &  x \in [\eta-r,\eta+r],\\
\displaystyle \widetilde{\sigma} \left(\frac{(\eta-r)x}{|x|} \right)
 = \widetilde{\sigma}(\eta-r), & x <\eta-r,\\
 \displaystyle \widetilde{\sigma} \left(\frac{(\eta+r)x}{|x|} \right)
 = \widetilde{\sigma}(\eta+r), & x > \eta+r,
\end{array}
\right.
$$
are bounded and globally Lipschitz continuous on~$\RR^*_+$, 
and clearly $\eps \mathfrak{b}(\cdot)$ converges uniformly to zero on~$\RR^*_+$ as~$\eps$ goes to zero.

Denote $\overline{Y}^\eps$ the solution to
$\D \overline{Y}^\eps_t = \eps\mathfrak{b}(\overline{Y}^\eps_t) \D t + \sqrt{\eps} \mathfrak{s}(\overline{Y}^\eps_t) \D W_t$
with $\overline{Y}^\eps_0 = y_0 >0$.
Then the two sequences $(\overline{Y}^\eps)_{\eps >0}$ and $(Y^\eps)_{\eps>0}$ are identical 
in~$B_r (\eta)$. 
Thus, for each $0 < \delta < y_0$ (small enough) there exist $\xi >0$ such that, for all $x\in B_\xi(y_0)$, 
$$
\PP \left[ \| Y^\eps - \Ss_2^{y_0} (f) \|_\infty >\delta, \| \sqrt{\eps} W- f \|_\infty \le \zeta \right]
= \PP \left[ \|\overline{Y}^\eps - \Ss_2^{y_0}(f) \|_\infty > \delta,\| \sqrt{\eps} W - f\|_\infty \le \zeta\right],
$$
for all $f\in \overline{\Hh}$ s.t. $\Lambda (f) \le \lambda$, with $\zeta, \lambda >0$ fixed.
Hence, for each $R, \lambda>0$ and $\delta >0$ small enough, there exist $\zeta, \xi, \eps_0 >0$ such that, for all $f\in \overline{\Hh}$ with $\Lambda(f) \le \lambda$, $x\in B_\xi (y_0)$, $\eps \le \eps_0$,
$$
\PP \left[ \|Y^\eps - \Ss_2^{y_0} (f)\|_\infty >\delta, \| \sqrt{\eps} W -f\|_\infty \le \zeta \right]
\le \exp\left\{-\frac{R}{\eps}\right\}
$$
holds from~\cite[Proposition 2.15]{P07} and so~\cite[Theorem 2.7]{P07} is satisfied here as well.
Finally, large deviations follow from the same reasoning as in the proof of~\cite[Theorem 2.9]{P07}.

\subsection{Proof of Theorem~\ref{thm:LDP_X}}\label{proof:LDP_X}
To obtain a large deviations principle for~$X^\eps$, a large deviations principle for the rescaled process $\mathrm{X}^\eps :=(X^\eps, Y^\eps)$ needs to be proved. 
This is
$$
\D
\mathrm{X}^\eps_t
= \eps\mathrm{b} (\mathrm{X}^\eps_t) \D t + \sqrt{\eps} \mathrm{a} (\mathrm{X}^\eps_t) \D W_t, 
$$
with initial condition 
$\mathrm{X}^\eps_0 
:= \mathrm{x}_0
=
\begin{pmatrix}
\log s_0 \\
y_0
\end{pmatrix} 
$
and the maps $\mathrm{b}, \mathrm{a}: \RR^*_+ \to \RR^2$ defined as
$$
\mathrm{b}(\mathrm{X}^\eps_t) =
\begin{pmatrix}
-\half\sigma^2(Y^\eps_t) \\
b(Y^\eps_t)
\end{pmatrix}
\qquad\text{and}\qquad
\mathrm{a}(\mathrm{X}^\eps_t) = 
\begin{pmatrix}
\sigma(Y^\eps_t) \\
\widetilde{\sigma}(Y^\eps_t)
\end{pmatrix}.
$$
These two maps are both locally Lipschitz continuous on~$\RR\times\RR^*_+$.
Solving the controlled ODE for $Y^\eps$ is sufficient to solve the controlled ODE for the process $\mathrm{X}^\eps$. Using the proof of Proposition~\ref{pp:LDP_Y}, 
for $\mathrm{f}:= (f,f)$ with $f\in \overline{\Hh}$, 
the controlled ODE $\dot{\mathrm{g}}_t =\dot{f}_t \mathrm{a}(\mathrm{g}_t) $,
with $\mathrm{g}_0 = \mathrm{x}_0$
has a solution $\mathrm{g} = \Ss^{\mathrm{x}_0}(f)$ on~$[0,T]$.
For $y_0 > y_\sigma$, the solution $\Ss_2^{y_0}$ is strictly positive and~$\Ss^{\mathrm{x}_0} (f)$ exists on~$[0,T]$ for all $f\in \overline{\Hh}$ and $\mathrm{x}_0 \in \RR\times\RR^*_+$. In this case, $\overline{\Hh}$ boils down to the Cameron-Martin space.
For $y_0 < y_\sigma$, Condition~\eqref{eq:Condition_f} ensures that~$\Ss_2^{y_0}$ is positive.
Applying~\cite[Theorem~2.9]{P07}, 
the sequence~$\mathrm{X}^\eps$ then satisfies a large deviations principle on $\Cc([0,T], \RR\times\RR^*_+)$ as~$\eps$ tends to zero, with speed~$\eps$ and rate function
$$
\Ir^{Y,X} (\mathrm{g})
:= \inf \left\{\Lambda (f), f\in \overline{\Hh}, \Ss^{\mathrm{x}_0} (f) = \mathrm{g} \right\}.
$$
To obtain a large deviations principle for the $\log$-stock price~$X^\eps$, we apply the Contraction Principle~\cite[Theorem 4.2.1]{DZ}, since the projection on the first component is continuous.

\subsection{Proof of Corollary~\ref{cor:SmallTimeIV}}
\label{cor:SmallTimeIV_proof}
We prove the lower and upper bounds separately, which turn out to be equal. 
For simplicity, we introduce the following notation, for all $k \neq 0$:
$$
 \widetilde{\mathrm{I}}^X (k)
 = 
 \left\{ \begin{array}{ll}
\displaystyle \inf _{y \geq k} \mathrm{I}^X (g)\rvert_{g(1)= y},
& \text{if }k>0,\\
\displaystyle \inf _{y \leq k} \mathrm{I}^X (g)\rvert_{g(1)= y},
& \text{if }k<0.
\end{array}
\right.
$$
Assuming that the rate function is continuous\footnote{Unless the rate function is available in closed form, it is hard to check for continuity.
This was done directly 
for the Heston model in~\cite{FJ09Heston} and in~\cite[Corollary 4.10]{FZ17} for a simplified rough volatility model.
The most general related statement is available in~\cite{FGP21} based on non-degeneracy assumptions.}, $\lim_{t \downarrow 0} t \log \PP \left[S_t > \E^k \right] = 
-   \widetilde{\mathrm{I}}^X (k)$.
We only consider $k>0$, the other case being symmetric.
The proof of this identity is similar  to that of ~\cite[Corollary 4.13, Appendix C]{FZ17}.
\begin{itemize}
\item[-]  For any $\delta >0$,
the inequality $\EE [(S_t - \E^k)_{+}]
\ge k \E^k \delta \PP 
[S_t > \E^{k(1+\delta)}]$ and Theorem~\ref{thm:LDP_X}, together with the continuity of the rate function,  then imply
$$
\liminf_{t \downarrow 0} t \log \EE\left[\left(S_t - \E^k\right)_{+}\right]
\ge \liminf_{t \downarrow 0} 
\left\{ t(k + \log k + \log \delta) + t\log \PP \left[S_t > \E^{k(1+\delta)}\right] \right\}
= -  \widetilde{\mathrm{I}}^X (k(1+\delta)).
$$
Take $\delta \downarrow 0$, by continuity of $ \widetilde{\mathrm{I}}^X (k)$, we obtain the desired lower bound.
\item[-]  To establish the desired upper bound, we note that for any $q>1$, we have
$$
\EE  \left[\left(S_t - \E^k\right)_{+}\right] 
\le \EE\left[\left(S_t - \E^k \right)_+^q \right]^{1/q} 
\PP \left[S_t \ge \E^k \right]^{1-1/q}.
$$
and therefore
$t \log \EE[(S_t - \E^k)_{+}]
\le \frac{t}{q} \log\EE[S_t^q] + t (1-\frac{1}{q}) \log \PP[S_t \ge \E^k]$.
From Theorem~\ref{Thm_behaviour}, for $y_\sigma \le \min \{y_0,1\}$, the process $(Y_t)_{t \in [0,T]}$ remains in $(y_\sigma, \infty)$. 
The map~$\sigma$ is bounded on $(y_\sigma, \infty)$, in particular $0 \leq \sigma(y) \leq \alpha/\beta$, and thus adapting the arguments in~\cite[proof of Corollary 1.2]{FJ11}, we have $\limsup_{t \downarrow 0} \frac{t}{q} \log\EE[S_t^q] \le  0$.
Indeed, exploiting H\" older inequality and the closed-form formula for the exponential moments of a Gaussian random variable, we have
\begin{align*}
    \EE[S_t^q]
    & =s_0^q \EE\left[\E^{q(X_t-x_0)}\right] \\
    & =s_0^q \EE\left[\exp\left\{q\left(-\half \int_0^t \sigma(Y_s)^2\D s + \int_0^t \sigma(Y_s)\D W_s \right)\right\}\right] \\
    & \leq s_0^q \EE\left[\exp\left\{-q\int_0^t \sigma(Y_s)^2\D s\right\}\right]^{\half} \EE\left[\exp\left\{2q\int_0^t \sigma(Y_s)\D W_s \right\}\right]^{\half}\\
    & \leq s_0^q \exp\left\{\frac{4q^2}{4}\mathbb{V}\left[\int_0^t \sigma(Y_s)\D W_s \right]\right\}
    \leq s_0^q \exp\left\{q^2\int_0^t \EE\left[\sigma(Y_s)^2\right]\D s \right\}
    \leq s_0^q \exp\left\{\frac{\alpha^2q^2}{\beta^2}t \right\},
\end{align*}
which yields
\begin{align*}
    \limsup_{t \downarrow 0}\frac{t}{q} \log\EE[S_t^q] 
    \leq  \limsup_{t \downarrow 0}  t \log \left(s_0^q \exp\left(\frac{\alpha^2q^2}{\beta^2}t \right) \right)
    \leq  \limsup_{t \downarrow 0}  t \left\{q x_0 + \frac{\alpha^2q^2}{\beta^2}t \right\}
    =0.
\end{align*}
Therefore, for fixed $q>1$, we have 
$
\limsup_{t \downarrow 0} t \log \EE [(S_t - \E^k)_{+} ]
\le  -(1-\frac{1}{q})  \widetilde{\mathrm{I}}^X (k)$.
Taking~$q$ to infinity yields the desired upper bound.
\end{itemize}


\end{document}